\documentclass[a4paper,11pt]{article}

\usepackage[utf8]{inputenc}
\usepackage[margin=1.2in]{geometry}
\usepackage{tikz} 
\usepackage{physics}
\usepackage{amsfonts}
\usepackage{amsthm}
\usepackage{amsmath}
\usepackage{algorithm}
\usepackage{algorithmic}
\usepackage{enumerate}
\usepackage{enumitem}
\usepackage{amssymb}
\usepackage{pdflscape}

\usepackage[affil-it]{authblk}
\usetikzlibrary{decorations.pathreplacing}
\usetikzlibrary{decorations.pathreplacing,angles,quotes}

\newtheorem{theorem}{Theorem}[section]

\newtheorem{lemma}{Lemma}[section]
\newtheorem{definition}{Definition}[section]

\newcommand{\ANNASCOMMENT}[1]{}

\newcounter{list}
\stepcounter{list}

\title{Information Theoretically Secure Hypothesis Test for Temporally Unstructured Quantum Computation}

\author[1]{Daniel Mills}
\author[1,2]{Anna Pappa}
\author[1,3]{Theodoros Kapourniotis}
\author[1,4]{Elham Kashefi}
\affil[1]{School of Informatics, University of Edinburgh, UK}
\affil[2]{Department of Physics, University College London, UK}
\affil[3]{Department of Physics, University of Warwick, UK}
\affil[4]{LIP6, CNRS, Pierre et Marie Curie University, Paris, France}
\date{}

\begin{document}

	\maketitle
	
	\begin{abstract}			
        \noindent	We propose a new composable and information-theoretically secure protocol to verify that a server has the power to sample from a sub-universal quantum machine implementing only commuting gates. By allowing the client to manipulate single qubits, we exploit properties of Measurement based Blind Quantum Computing to prove security against a malicious Server and therefore certify quantum supremacy without the need for a universal quantum computer. 
	\end{abstract}

\section{Introduction}
\label{sec:introduction}

Quantum computers are believed to be efficient for simulating quantum systems \cite{Simulating Physics with Computers,Quantum Simulation} and have been shown to have many other applications \cite{Quantum Computing and Quantum Information}. Protocols demonstrating the power of quantum computers include Shor's algorithm for prime factorisation \cite{Poly-Time Algorithms for Prime Factorisation and Discrete Logarithms on a Quantum Computer}, Grover's algorithm for unstructured search \cite{A fast quantum mechanical algorithm for database search}, and the BB84 protocol for public key exchange \cite{Quantum cryptography: Public key distribution and coin tossing}. 

That said, it may be some time before a large scale universal quantum computer capable of demonstrating the computational power of these protocols is built. In the meantime several intermediate, non-universal models of quantum computation, like the one clean qubit model \cite{Power of One Bit of Quantum Information,Hardness of Classically Simulating the One-Clean-Qubit Model} and the boson sampling model \cite{An Introduction to Boson-Sampling}, have been developed and may prove easier to implement. The \emph{Instantaneous Quantum Poly-time} (IQP) machine \cite{Temporally_Unstructured_Quantum_Computation} is another such non-universal model with significant practical advantages \cite{Fault-tolerant computing with biased-noise superconducting qubits: a case study, Architectures for quantum simulation showing quantum supremacy}. In spite of the fact that IQP uses only commuting gates (in contrast to the non-commuting gate set needed for universal computations), it is believed to remain hard to classically simulate \cite{Classical Simulation of Commuting Quantum Computations Implies Collapse of the Polynomial Hierarchy,Average-case complexity versus approximate simulation of commuting quantum computations} even in a noisy environment \cite{Achieving quantum supremacy with sparse and noisy commuting quantum computations}. Hence, providing evidence that a machine can perform hard IQP computations would be a proof of its quantum supremacy.

In \cite{Temporally_Unstructured_Quantum_Computation}, the authors present a \emph{hypothesis test} that can be passed only by devices capable of efficiently simulating IQP machines, providing the aforementioned evidence of the capability to perform hard IQP computations. The client in that work is purely classical, however computational assumptions (conjecturing the hardness of finding hidden sub-matroids) were required for the security of the test against a malicious server. In the present work, by providing a suitable implementation of the IQP machine in the setting of Measurement Based Quantum Computing (MBQC) \cite{A one-way quantum computer,Measurement-based quantum computation on cluster states}, we are able to use tools from quantum cryptography (e.g. blind quantum computing \cite{Universal Blind Quantum Computation,Unconditionally Verifiable Blind Quantum Computation}) to develop an information-theoretically secure hypothesis test. To do so, we need to empower the client with minimal quantum capabilities such as those required in standard Quantum Key Distribution schemes. 

The structure of this work is as follows. In Section \ref{sec:Preliminaries}, we formally introduce the IQP machine and develop an implementation of it in MBQC that is more suitable for our blind delegated setting than previous ones \cite{Temporally_Unstructured_Quantum_Computation,Measurement-based classical computation}. In Section \ref{sec:Blind IQP} we derive a delegated protocol for IQP computations that keeps the details of the computation hidden from the device performing it, and prove information-theoretic security in a composable framework. Finally in Section \ref{sec:hypothesis test} we develop our hypothesis test for quantum supremacy, which a limited quantum client can run on an untrusted Server.

	
\section{Preliminaries}
\label{sec:Preliminaries}
\subsection{X-programs}
\label{subsec:Xprograms}

The IQP machine introduced in \cite{Temporally_Unstructured_Quantum_Computation}, is defined by its capacity to implement $X$-programs. 
\begin{definition}
    An \emph{$X$-program} consists of a Hamiltonian comprised of a sum of products of $X$ operators on different qubits, and $\theta\in[0,2\pi]$ describing the action for which it is applied.  The $i$-th term of the sum has a corresponding vector $\mathbf{q}_{i}$, called a \emph{program element}, which defines on which of the $n_p$ input qubits, the product of $X$ operators, which constitute that term, act. $\mathbf{q}_i$ has 1 in the $j$-th position when $X$ is applied on the $j$-th qubit.
    
    As such, we can describe the $X$-program using $\theta$ and a poly-size list of $n_a$ vectors $ \mathbf{q}_i \in \left\{ 0 , 1 \right\} ^{n_{p}}$ or, if we consider the matrix $\mathbf{Q}$ which has as rows the program elements $\mathbf{q}_i,i=1,\dots,n_a$, simply by the pair $\left( \mathbf{Q} , \theta \right) \in \left\{ 0 , 1 \right\}^{n_{a} \times n_{p}} \times \left[ 0 , 2 \pi \right]$.
\end{definition}

Applying the $X$-program discussed above to the computational basis state $\ket{0^{n_p}}$ and measuring the result in the computational basis allows us to see an $X$-program as a quantum circuit with input $\ket{0^{n_p}}$, comprised of gates diagonal in the Pauli-X basis, and classical output. Using the random variable $X$ to represent the distribution of output samples, the probability distribution of outcomes $\widetilde{x} \in \{0,1\}^{n_{p}}$ is:
\begin{equation}
	\label{equ:IQP probability distribution}
	\mathbb{P} \left( X = \widetilde{x} \right) = \left| \bra{\widetilde{x}} \exp \left( \sum_{i=1}^{n_a} i \theta \bigotimes_{j: \mathbf{Q}_{ij} = 1} X_{j}\right) \ket{0^{n_{p}}}  \right|^2
\end{equation}
Note that the $i$ not used as an index is the imaginary unit.

\begin{definition}
	\label{def:IQP oracle}
	Given some X-program, an \emph{IQP machine} is any computational method capable of efficiently returning a sample $\widetilde{x} \in \{0,1\}^{n_{p}}$ from the probability distribution \eqref{equ:IQP probability distribution}.
\end{definition}


\subsection{IQP In MBQC}
\label{subsec:IQP in MBQC}

We present an implementation of a given $X$-program in MBQC that will be used later in our protocol design. First notice that using the equality below: 
\begin{equation*}
	\label{equ:X-prog unitary with Z-prog}
	\exp \left( \sum_{i=1}^{n_a} i \theta \bigotimes_{j: \mathbf{Q}_{ij} = 1} X_{j}\right) = H_{n_p} \left( \prod_{i=1}^{n_a}\exp \left( i \theta \bigotimes_{j: \mathbf{Q}_{ij} = 1} Z_{j}\right) \right) H_{n_p}
\end{equation*}
equation \eqref{equ:IQP probability distribution} can be rewritten as:
\begin{equation}
	\label{equ:Z-gate X-prog distribution}
	\mathbb{P} \left( X = \widetilde{x} \right) = \left| \left( \bra{\widetilde{x}} H_{n_p} \right) \left( \prod_{i=1}^{n_a} \exp \left( i \theta \bigotimes_{j: \mathbf{Q}_{ij} = 1} Z_{j}\right) \right) \ket{\mathbf{+}^{n_p}}  \right|^2
\end{equation}
For any given $i$, we now show the following lemma.
\begin{lemma}
	\label{lemma:X-prog circuit}
	The circuit of Figure \ref{fig:z-prog circuit} implements the unitary:
	\begin{equation}
	    \label{equ:single Z-prog unitary term}
	    \exp \left( i \theta \bigotimes_{j : \mathbf{Q}_{ij} = 1} Z_{j}\right)
    \end{equation}
\end{lemma}

\begin{figure}
	\centering
	\begin{tikzpicture}[scale = 0.7]
		\draw[very thick] (0,6) node[anchor=east] {$\tilde{p}_1$} -- (7,6);
		\node at (0,5) {$\vdots$};
		\draw[very thick] (0,4) node[anchor=east] {$\tilde{p}_{\# i}$} -- (5,4);
		\draw[very thick] (0,3) node[anchor=east] {$\tilde{p}_{\# i+1}$} -- (9,3);
		\node at (0,2) {$\vdots$};
		\draw[very thick] (0,1) node[anchor=east] {$\tilde{p}_{n_{p}}$} -- (9,1);
		\draw[very thick] (0,0) node[anchor=east] {$\ket{+}$} -- (4,0);
		
		\filldraw[very thick] (1,0) circle (3pt) -- (1,6) circle (3pt);
		
		\node at (2,5) {$\ddots$};
		
		\filldraw[very thick] (3,0) circle (3pt) -- (3,4) circle (3pt);
		
		\draw[thick] (4,-0.4) rectangle (5,0.4);
		\draw (4.1,-0.1) .. controls (4.3,0.2) and (4.7,0.2) .. (4.9,-0.1);
		\draw[thick, ->] (4.5, -0.2) -- (4.8, 0.3);
		
		\draw[very thick] (5,-0.1) -- (7.6,-0.1) -- (7.6,5.5);
		\draw[very thick] (5,0.1) -- (5.4,0.1) -- (5.4,3.5);
		\draw[very thick] (5.6,3.5) -- (5.6,0.1) -- (7.4,0.1) -- (7.4,5.5);
		
		\draw[very thick] (5,3.5) rectangle (6,4.5) node[pos = 0.5] {$Z$};
		\draw[very thick] (6,4) -- (9,4);
		
		\node at (6.5,5) {\reflectbox{$\ddots$}};
		
		\draw[very thick] (7,5.5) rectangle (8,6.5) node[pos = 0.5] {$Z$};
		\draw[very thick] (8,6) -- (9,6);
		
	\end{tikzpicture}
	\caption{The circuit implementing Expression \eqref{equ:single Z-prog unitary term}. The input qubits $\{p_j\}_{j=1}^{n_p}$ are rearranged so that if $\# i$ is the Hamming weight of row $i$ of matrix $\mathbf{Q}$, then for $k=1,\dots,\#i$ each $\tilde{p}_k$ corresponds to one $p_j$ such that $\mathbf{Q}_{ij}=1$ and for $k=\#i+1,\dots,n_p$ they correspond to the ones such that $\mathbf{Q}_{ij}=0$.  The ancillary qubit measurement is in the basis $\left\{ \ket{0_{\theta}} , \ket{1_{\theta}} \right\}$ defined in expression \eqref{equ:IQP measuremnt basis}.}
	\label{fig:z-prog circuit}
\end{figure}
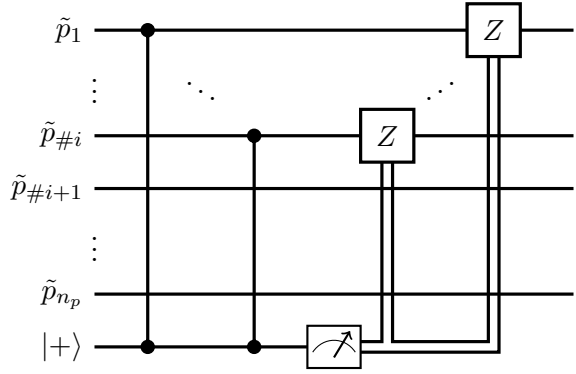

\begin{proof}

To prove this statement, we will prove that the effect of Figure \ref{fig:z-prog circuit} and expression \eqref{equ:single Z-prog unitary term} is the same on all inputs. Without loss of generality, we can consider only computational basis input states $\ket{p}=\ket{p_1}\dots\ket{p_{n_p}}$, $p_{j} \in \left\{ 0 , 1 \right\}$. Since the operation that we perform is linear, the result then follows for all inputs.

Notice that, representing the $n_{p}$-qubit identity operator by $\mathbb{I}_{n_{p}}$, we can rewrite Expression \eqref{equ:single Z-prog unitary term} as: 
	\begin{equation}
		\label{equ:single Z-prog unitary term rewritten}
		 \cos{\theta} \mathbb{I}_{n_{p}}+ i \sin{\theta} \bigotimes_{j : \mathbf{Q}_{ij} = 1} Z_{j} 
	\end{equation}
The above operator on $\ket{p}$ has two possible outcomes:
	\begin{enumerate}  
		\item For the $j \in \left\{ 1,\dots,n_p \right\}$ such that $\mathbf{Q}_{ij}=1$, if the number of $\ket{p_j}=\ket{1}$ is even, then there will be a phase change of $\cos{\theta} + i \sin{\theta}$, as the $\bigotimes_{j : \mathbf{Q}_{ij} = 1} Z_{j}$ operator will extract an even number of negatives. 
		\item For the  $j$'s ($j=1,\dots,n_p$) such that $\mathbf{Q}_{ij}=1$, if the number of $\ket{p_j}=\ket{1}$ is odd, then the phase change will be $\cos{\theta} - i \sin{\theta}$.
	\end{enumerate}
Hence, depending on the parity of $\ket{p}$ in the positions where $\mathbf{Q}_{ij}=1$, the effect is to produce one of the two states:
	\begin{equation}
		\label{equ:single Z-prog unitary term operating on basis}
		\left( \cos{\theta} \pm i \sin{\theta} \right) \ket{p} = e^{ \pm i \theta }\ket{p}
	\end{equation}

We now show the effect of the circuit in Figure \ref{fig:z-prog circuit} is the same as the operator in expression \eqref{equ:single Z-prog unitary term rewritten}. For ease of readability, in Figure  \ref{fig:z-prog circuit} we consider a permutation of the states $\ket{\tilde{p}_1},\dots,\ket{\tilde{p}_{\#i}},\dots,\ket{\tilde{p}_{n_p}}$ such that the first $\#i$ qubits are the ones for which the value in the corresponding position in the program element is 1.
	
The action of the controlled-Z gates is to check the parity of $\ket{1}$'s in the input as each appearance of a $\ket{1}$ will flip the bottom \emph{ancillary} qubit between the states $\ket{+}$ and $\ket{-}$. After the action of all controlled-Z operators, we have the state $\ket{p} \ket{+}$ if there is an even number of $\ket{\tilde{p}_k}=\ket{1}$ for $k=1,\dots,\#i$ and $\ket{p} \ket{-}$ if this number is odd. Making a measurement of the ancillary qubit in the basis:
\begin{equation}
	\label{equ:IQP measuremnt basis}
	\left\{ \ket{0_{\theta}} , \ket{1_{\theta}} \right\}= \left\{ \frac{1}{\sqrt{2}} \left(  e^{-i\theta}\ket{+} + e^{i\theta}\ket{-} \right) ,\frac{1}{\sqrt{2}} \left(e^{-i\theta} \ket{+} - e^{i\theta}\ket{-} \right)  \right\} 
\end{equation}  
leaves us with one of the two states $\pm e^{ - i \theta}\ket{p}$ in the odd parity case and with the state $e^{ i \theta} \ket{p}$ in the even parity case. The negative sign preceding the exponential term in the odd parity case comes from measuring the state $\ket{1_{\theta}}$ (a measurement outcome of $1$) and the positive sign comes from measuring $\ket{0_{\theta}}$.

In the case of a measurement outcome $1$, we then apply $Z$ operators to all unmeasured qubits to ensure that the resulting states are as in expression \eqref{equ:single Z-prog unitary term operating on basis} and with the same dependency of the sign on the parity of $\ket{p}$.
\end{proof}

We now consider generating the full distribution of equation \eqref{equ:Z-gate X-prog distribution} using measurement based quantum computing \cite{A one-way quantum computer,Measurement-based quantum computation on cluster states}. An MBQC computation consists of a graph describing the pattern of entanglement amongst the qubits in a state, a measurement pattern describing the order of measurements of qubits in that state, and a set of corrections on later measurements which can depend on the outcomes of previous ones. We now identify all of these components of an MBQC computation in the case of an IQP computation.


\begin{lemma}
	\label{lem:IQP graph design}
	A graph and measurement pattern can always be designed to simulate an $X$-program efficiently.
\end{lemma}

\begin{proof}
	Producing the distribution in Eq. \eqref{equ:Z-gate X-prog distribution} can be achieved by inputting the state $\ket{+^{n_{p}}}$ into a circuit made from composing circuits like the one in Figure \ref{fig:z-prog circuit} (one for each term of the product in Eq. \eqref{equ:Z-gate X-prog distribution}) and measuring the result in the Hadamard basis. The $Z$ corrections commute with the controlled-$Z$ operations and therefore they can be moved to the end of the new, larger circuit.
	
	Because there is no dependency between the measurements, they can be performed in any order or even simultaneously. The $Z$ corrections, conditional on the measurement outcomes of the ancillary bits, can then be implemented via bit flips.
\end{proof}

A formal description of the protocol described in these proofs can be found in Algorithm \ref{alg:x-prog in MBQC} of the Appendix. We introduce some further terminology which is used in that algorithm and in the remainder of this work. 

The reader will notice that the entanglement pattern used in Algorithm \ref{alg:x-prog in MBQC} and implicit in the proof of Lemma \ref{lem:IQP graph design} is that of an \emph{undirected bipartite graph}, which we will refer to as an IQP graph.

\begin{definition}
    An \emph{undirected bipartite graph}, which we refer to as an \emph{IQP graph}, consists of a bipartition of vertices into two sets $P$ and $A$ of cardinality $n_p$ and $n_a$ respectively. We may represent such a graph by $\mathbf{Q} \in \left\{ 0 , 1 \right\}^{n_{a} \times n_{p}}$. An edge exists in the graph when $\mathbf{Q}_{ij}=1$, for $i=1,\dots,n_a$ and $j=1,\dots,n_p$. We call the set $P$ \emph{primary vertices} and the set $A$ \emph{ancillary vertices}.
\end{definition}

By referring to the bottom qubit of Figure \ref{fig:z-prog circuit} as the ancillary qubit and the others as primary qubits we understand why this type of graph is relevant and how the $X$-program matrix $\mathbf{Q}$, interpreted and a bipartite graph, exactly describes the entanglement pattern. Throughout this work, we refer to $\mathbf{Q}$ interchangeably as a matrix corresponding to an $X$-program and a graph and the reader may wish to direct their attention to Figure \ref{fig:bipartite graph} for an example. 


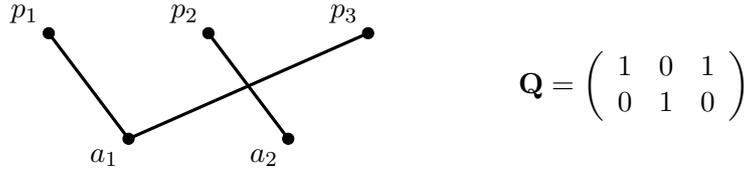
\begin{figure}
	\centering
	\begin{tikzpicture}[scale = 0.7]
		\filldraw (1.5,0) circle (3pt) node[anchor = north east] {$a_{1}$};
		\filldraw (4.5,0) circle (3pt) node[anchor = north east] {$a_{2}$};
		
		\filldraw (0,2) circle (3pt) node[anchor = south east] {$p_{1}$};
		\draw[very thick] (0,2) -- (1.5,0);
		
		\filldraw (3,2) circle (3pt) node[anchor = south east] {$p_{2}$};
		\draw[very thick] (3,2) -- (4.5,0);
		
		\filldraw (6,2) circle (3pt) node[anchor = south east] {$p_{3}$};
		\draw[very thick] (6,2) -- (1.5,0);
		
		\node at (11,1) {$\mathbf{Q} =   \left( 
		                                    \begin{array}{ccc}
                                                1 & 0 & 1 \\
                                                0 & 1 & 0
                                            \end{array}
                                        \right)$};
	\end{tikzpicture}
	\caption{An example of an IQP graph described by matrix $\mathbf{Q}$. Here, $n_{p} = 3$ and $n_{a} = 2$ while the partition used is $P = \left[ p_1 , p_2 , p_3 \right]$ and $A = \left[ a_1 , a_2 \right]$.}
	\label{fig:bipartite graph}
\end{figure}

\section{Blind Delegated IQP Computation}
\label{sec:Blind IQP}

The next step towards our method for verifying IQP machines is to build a method for blindly performing an IQP computation in a delegated setting. We consider a Client with limited quantum power delegating an IQP computation to a powerful Server. The novel method that we use in this work is to keep the $X$-program secret by not revealing the quantum state used. The intuition behind the method used to perform this hiding is that the Client will ask the Server to produce a quite general quantum state and then move from that one to the one that is required for the computation. If this is done in a blind way then the Server only has some knowledge of the general starting state from which any number of other quantum states may have been built. Hence, there are two key problems that to be addressed in the following subsections:

\begin{enumerate}
	\item \label{pt:BIQP problem 1} How to move from a general quantum state to a specific one representing an IQP computation.
	\item \label{pt:BIQP problem 2} How to do so secretly in a delegated setting.
\end{enumerate}


\subsection{Break and Bridge}
\label{sec:Break, Bridge Operators}

\ANNASCOMMENT{\begin{figure}
    \centering
	\begin{tikzpicture}[scale = 0.7]
	    \draw[very thick] (0,0) -- (1,1) -- (2,0) -- (3,1);
	\end{tikzpicture}
\end{figure}}

The break and bridge operations on a graph $\widetilde{G}=(\widetilde{V},\widetilde{E})$, with vertex set $\widetilde{V}$ and edge set $\widetilde{E}$, which were introduced in \cite{Unconditionally Verifiable Blind Quantum Computation, Multi-party entanglement in graph states}, are exactly those necessary to solve the `how to move' element of problem \ref{pt:BIQP problem 1}. 

\begin{definition}
	\label{def:bridge and break}
	The \emph{break} operator acts on a vertex $v \in \widetilde{V}$ of degree 2 in a graph $\widetilde{G}$. It removes $v$ from $\widetilde{V}$ and also removes any edges connected to $v$ from $\widetilde{E}$. 
  
	The \emph{bridge} operator acts also on a vertex $v \in \widetilde{V}$ of degree 2 in a graph $\widetilde{G}$. It removes $v$ from $\widetilde{V}$, removes any edges connected to $v$ from $\widetilde{E}$ and adds a new edge between the neighbours of $v$.
\end{definition}

Figure \ref{fig:bridge and break chain} gives an example of multiple applications of the bridge and break operators. Once this is translated from a graph theoretic idea to an operation on quantum states, we will have address the `how to move' component of problem \ref{pt:BIQP problem 1}. 

The \emph{extended IQP graphs}, which we define now, is the `general quantum state' also mentioned in problem \ref{pt:BIQP problem 1}.

\ANNASCOMMENT{\begin{definition}
To refer to the resulting graph we introduce the following terminology.
	\label{def:br-sub-graph}

    Given a graph, $\widetilde{G}$, a new graph, $G$ called a \emph{br-sub-graph}, is obtained from $\widetilde{G}$ by applying consecutive break or bridge operations to $\widetilde{G}$. We impose the condition that the vertices to which break and bridge operations are applied must not neighbour each other in the graph $\widetilde{G}$.
	
	A \emph{br-sup-graph}, $\widetilde{G}$, of a graph, $G$, is such that $G$ is a br-sub-graph of $\widetilde{G}$.
\end{definition}}

\begin{figure}
	\centering
	\begin{tikzpicture}[scale = 0.7]
		\filldraw (0,0) circle (3pt);
		\filldraw (0,1) circle (3pt);
		\filldraw (0,2) circle (3pt);
		\filldraw (1,0) circle (3pt);
		\filldraw (1,1) circle (3pt);
		\filldraw (1,2) circle (3pt);
		\filldraw (2,0) circle (3pt);
		\filldraw (2,1) circle (3pt);
		\filldraw (2,2) circle (3pt);
		
		\draw[very thick] (0,0) -- (0,1);
		\draw[very thick] (0,1) -- (1,0);
		\draw[very thick] (2,0) -- (2,1);
		\draw[very thick] (2,1) -- (2,2);
		\draw[very thick] (0,2) -- (1,2);
		\draw[very thick] (1,0) -- (1,1);
		\draw[very thick] (1,1) -- (2,2);
		\draw[very thick] (1,2) -- (2,2);
		
		\draw[very thick , ->] (2.5,1) -- (3.5,1);
		
		\filldraw (4,0) circle (3pt);
		\filldraw (4,1) circle (3pt);
		\filldraw (4,2) circle (3pt);
		\filldraw (5,0) circle (3pt);
		\filldraw (5,1) circle (3pt);
		\filldraw (5,2) circle (3pt);
		\filldraw (6,0) circle (3pt);
		\filldraw (6,2) circle (3pt);

		\draw[very thick] (4,0) -- (4,1);
		\draw[very thick] (4,1) -- (5,0);
		\draw[very thick] (4,2) -- (5,2);
		\draw[very thick] (5,0) -- (5,1);
		\draw[very thick] (5,1) -- (6,2);
		\draw[very thick] (5,2) -- (6,2);
		\draw[very thick] (6,2) -- (6,0);
		
		\draw[very thick , ->] (6.5,1) -- (7.5,1);
		
		\filldraw (8,0) circle (3pt);
		\filldraw (8,1) circle (3pt);
		\filldraw (8,2) circle (3pt);
		\filldraw (9,0) circle (3pt);
		\filldraw (9,1) circle (3pt);
		\filldraw (10,2) circle (3pt);
		\filldraw (10,0) circle (3pt);

		\draw[very thick] (8,0) -- (8,1);
		\draw[very thick] (8,1) -- (9,0);
		\draw[very thick] (10,0) -- (10,2);
		\draw[very thick] (9,0) -- (9,1);
		\draw[very thick] (9,1) -- (10,2);
		
	\end{tikzpicture}
	\caption{An example of a sequence of one bridge and one break operation.}
	\label{fig:bridge and break chain}
\end{figure}
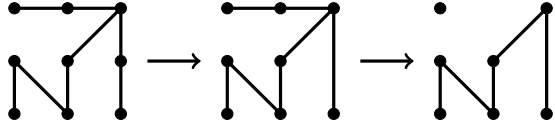

\begin{definition}
    An \emph{extended IQP graph} is represented by $\widetilde{\mathbf{Q}} \in \left\{ -1 , 0 , 1 \right\}^{n_{a} \times n_{p}}$. The vertex set contains $A = \left\{ a_{1} , ... , a_{n_{a}} \right\}$ and $P = \left\{ p_{1} , ... , p_{n_{p}} \right\}$ while $\widetilde{\mathbf{Q}}_{ij}=0$ and $\widetilde{\mathbf{Q}}_{ij}=1$ has the same implications, regarding the connections between these vertices, as in IQP graphs. 
    
    We interpret $\widetilde{\mathbf{Q}}_{ij}=-1$ as the existence of an intermediary vertex $b_k$ between vertices $p_j$ and $a_i$, and denote with $n_b$ the number of -1s in $\widetilde{\mathbf{Q}}$. As such the vertex set also includes the \emph{bridge and break vertices} $B = \left\{ b_{1} , ... , b_{n_b} \right\}$ and the edge set includes edges between $b_{k}$ and $a_{i}$ as well as $b_{k}$ and $p_{j}$ when $\widetilde{\mathbf{Q}}_{ij}=-1$. To keep track of these connections we  define the surjective function $g$ for which $g \left( i , j \right) = k $ where $b_{k}$ is the intermediate vertex connected to $a_{i}$ and $p_{j}$. 
\end{definition}

An \emph{extended IQP graph} $\widetilde{\mathbf{Q}}$ can be built from an IQP graph $\mathbf{Q}$ by replacing any number of the entries of $\mathbf{Q}$ with $-1$. Throughout the remainder of this work we will use the tilde notation to represent an extended IQP graph $\widetilde{\mathbf{Q}}$ build from an IQP graph $\mathbf{Q}$ in this way. 

Figure \ref{fig:IQP extended graph} displays an example of an extended IQP graph. By applying a bridge operator to $b_{1}$ and a break operation to $b_{2}$ in $\widetilde{\mathbf{Q}}$ of Figure \ref{fig:IQP extended graph} we arrive at $\mathbf{Q}$ of Figure \ref{fig:bipartite graph}. It is in this sense that an extended IQP graph is `more general' that an IQP graph.

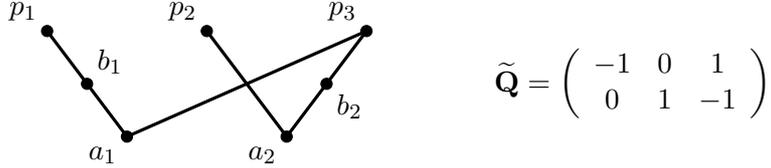
\begin{figure}
	\centering
	\begin{tikzpicture}[scale = 0.7]
		\filldraw (1.5,0) circle (3pt) node[anchor = north east] {$a_{1}$};
		\filldraw (4.5,0) circle (3pt) node[anchor = north east] {$a_{2}$};
		
		\filldraw (0,2) circle (3pt) node[anchor = south east] {$p_{1}$};
		\draw[very thick] (0,2) -- (1.5,0);
		
		\filldraw (0.75,1) circle (3pt) node[anchor = south west] {$b_{1}$};
		
		\filldraw (3,2) circle (3pt) node[anchor = south east] {$p_{2}$};
		\draw[very thick] (3,2) -- (4.5,0);
		
		\filldraw (6,2) circle (3pt) node[anchor = south east] {$p_{3}$};
		\draw[very thick] (6,2) -- (1.5,0);
		
		\draw[very thick] (6,2) -- (4.5,0);
		\filldraw (5.25,1) circle (3pt) node[anchor = north west] {$b_{2}$};
		
		\node at (11,1) {$\widetilde{\mathbf{Q}} =   \left( 
		                                    \begin{array}{ccc}
                                                -1 & 0 & 1 \\
                                                0 & 1 & -1
                                            \end{array}
                                        \right)$};
	\end{tikzpicture}
	\caption{An example of an extended  IQP graph described by matrix $\widetilde{\mathbf{Q}}$ with $(n_a,n_p,n_b)=(2,3,2)$, $P = \left[ p_1 , p_2 , p_3 \right]$ and $A = \left[ a_1 , a_2 \right]$. Two vertices $b_1$ and $b_2$ are introduced and  the function $g: \mathbb{Z}_{n_a\times n_p} \rightarrow \mathbb{Z}_{n_b}$ is defined as $g \left( 1 , 1 \right) = 1$ and $g \left( 2 , 3 \right) = 2$.}
	\label{fig:IQP extended graph}
\end{figure}


It is convenient to now introduce the following definition which allows us to use the graphs defined above to describe the entanglement pattern of quantum states.

\begin{definition}
    Consider a matrix $\mathbf{G} \in \left\{-1,0,1\right\}^{n_a\times n_p}$ and use function $g \left( i , j \right) = k$ to define index $k=1,\dots,n_b$ for the elements $\mathbf{G}_{ij}=-1$. The circuit $E_{\mathbf{G}}$ on $(n_a + n_p +n_b)$ qubits applies controlled-$Z$ operations between qubits $p_{j}$ and $a_{i}$ if $\mathbf{G}_{ij} = 1$ and, between qubits $b_{g(i,j)}$ and $a_{i}$, and, $b_{g(i,j)}$ and $p_{j}$, when $\mathbf{G}_{ij} = -1$.
\end{definition}

Using the above notation, the state built in Lemma \ref{lem:IQP graph design} is $E_{\mathbf{Q}} \ket{+}^{n_{a} + n_{p}}$. We refer to such a state, or $Z$ rotations there of, as an \emph{IQP state}. We will call states of the form $E_{\mathbf{Q}} \ket{+}^{n_{a} + n_{p}}$ or, again, their $Z$ rotations, as \emph{IQP extended state}

We can now state Lemma \ref{lem:state bridge and break correctness} which teaches us how to translate bridge and break operations from graph theoretical ideas into practical operations on quantum states. A similar lemma can be found in \cite{Unconditionally Verifiable Blind Quantum Computation}. 

\begin{lemma}
	\label{lem:state bridge and break correctness}
	Consider a quantum state $E_{\mathbf{Q}}\ket{\phi}$ where $\ket{\phi}$ is arbitrary. If $\widetilde{\mathbf{Q}}$ is an extended IQP graph built from $\mathbf{Q}$ then there exists a state $E_{\widetilde{\mathbf{Q}}}\ket{\psi}$, which can be transformed into the state $E_{\mathbf{Q}}\ket{\phi}$ through a sequence of Pauli-$Y$ basis measurements on qubits and local rotations around the Z axis on the unmeasured qubits through angles $\left\{ 0 , \frac{\pi}{2} , \pi , \frac{3 \pi}{2} \right\}$.
\end{lemma}

The detailed proof of Lemma \ref{lem:state bridge and break correctness}, which can be found in Appendix \ref{apen:Detailed proof of Lemma state bridge and break correctness}, shows us that we can create the following state.
\begin{equation}
	\label{equ:section final state}
	\prod_{k = 1}^{n_{b}} \left( S_{p_{j}}^{(-1)^{s^{b}_{k}+r^{b}_{k}}} \otimes S_{a_{i} }^{(-1)^{s^{b}_{k}+r^{b}_{k}}} \right)^{d^{b}_{k}} \left( Z_{p_{j}}^{r^{b}_{k}} \otimes Z_{a_{j}}^{r^{b}_{k}} \right)^{1-d_{k}} E_{\mathbf{Q}}\ket{\phi}
\end{equation}
where $p_j$ and $a_i$ are the primary and ancillary qubits connected to $b_k$ respectively. 

The operations performed to achieve this are measurements of the qubits corresponding to bridge and break vertices (which we call \emph{bridge and break qubits}) of $E_{\widetilde{\mathbf{Q}}}\ket{\psi}$ in the Pauli $Y$ basis. The quantity $s_k^{b}$ is the outcome of this measurement on qubit $b_{k}$ while the quantities $r^{b}_{k}$ and $d^{b}_{k}$ tell us that said qubit was initialised in the state $\ket{b_k} = Y^{r^{b}_{k}} \sqrt{Y}^{d^{b}_{k}} \ket{0}$. 

It is possible to perform an IQP computation using this method. Although the quantum state generated using this method would equal $E_{\mathbf{Q}} \bigotimes_{1}^{n_{a} + n_{p}} \ket{+}$ up to some $S$ corrections, these corrections may be accounted for by making corrections to the primary and ancillary measurement bases (see also the circuits in Figures \ref{fig:original graph generation circuit} and \ref{fig:bridge and break graph generation circuit} in Appendix \ref{sec:Pictorial Evolution of Algorithms in This Paper}). Algorithm \ref{alg:real IQP resource honest server distributed} of the Appendix uses the methods discussed to build an IQP state.

\subsection{The Protocol}
\label{sec:security proof}

We can now address problem \ref{pt:BIQP problem 2} of the introduction to this section. To do so we use the tools of the previous section to blindly create an IQP state at the Server side. What we wish is to construct the \emph{Ideal Resource} of Figure \ref{fig:ideal resource} which takes as input from the Client an IQP computation, $\left( \mathbf{Q} , \theta \right)$, and in return gives a classical output $\widetilde{x}$. If the Server is honest, then $\widetilde{x}$ comes from the distribution corresponding to $\left( \mathbf{Q} , \theta \right)$. If the Server is dishonest, then they can input some quantum operation $\mathcal{E}$ and some quantum state $\rho_{B}$ and force the output to the Client into the classical state $\mathcal{E}\left( \mathbf{Q} , \theta , \rho_{B} \right)$. We would like for the Server only to receive a IQP extended graph $\widetilde{\mathbf{Q}}$ which can be built from $\mathbf{Q}$, the distribution $\mathcal{Q}$ over the possible $\mathbf{Q}$ from which $\widetilde{\mathbf{Q}}$ could be built and $\theta$. Let us assume that this is public knowledge. 

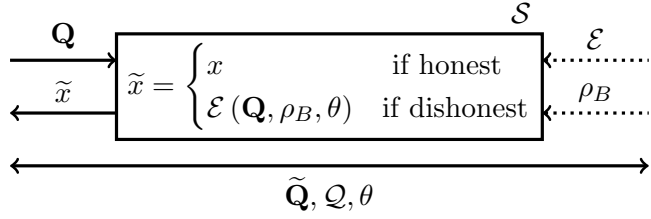
\begin{figure}
  \centering
  \begin{tikzpicture}[scale = 0.7]
     
    \draw[very thick] (0,0.5) rectangle (8,2.5) node[align = center, pos = 0.5] {$\widetilde{x} = \begin{cases} x & \text{ if honest} \\ \mathcal{E} \left( \mathbf{Q} , \rho_{B} , \theta \right) & \text{if dishonest}\end{cases}$} node[anchor = south east] {$\mathcal{S}$};
    
    \draw[very thick, ->] (-2,2) -- (0,2) node[anchor = south, pos = 0.5] {$\mathbf{Q}$};
    \draw[very thick, <-] (-2,1) -- (0,1) node[anchor = south, pos = 0.5] {$\widetilde{x}$};
    
    \draw[very thick, dotted, <-] (8,2) -- (10,2) node[anchor = south, pos = 0.5] {$\mathcal{E}$};
    \draw[very thick, dotted, <-] (8,1) -- (10,1) node[anchor = south, pos = 0.5] {$\rho_{B}$};
    
    \draw[very thick, <->] (-2,0) -- (10,0) node[anchor = north, pos = 0.5] {$\widetilde{\mathbf{Q}} , \mathcal{Q},\theta$};
  \end{tikzpicture}
  \caption{The ideal blind delegated IQP computation resource.}
     \label{fig:ideal resource}
\end{figure}


The proposed real communication protocol is described in detail by Algorithm \ref{alg:real IQP resource honest server} and graphically shown in Figure \ref{fig:real IQP computation resource}. The element of blindness is added to the work of Section \ref{sec:Break, Bridge Operators} and Algorithm \ref{alg:real IQP resource honest server distributed} by introducing some random rotations on the primary and ancillary qubits. These rotations are such that they can be corrected by rotating, in the same way, the measurement bases of those qubits, and therefore ensuring that the original IQP computation is being performed.

\begin{algorithm}
  \caption{Blind distributed IQP computation}
  \label{alg:real IQP resource honest server}
  \textbf{Public:} $\widetilde{\mathbf{Q}} , \mathcal{Q}$, $\theta$
  
  \textbf{Client input:} $\mathbf{Q}$ 
  
  \textbf{Client output:} $\widetilde{x}$ 
  
  \textbf{Protocol:}
  
  \begin{algorithmic}[1]
    \STATE The Client randomly generates $r^{p} , d^{p} \in \left\{ 0 , 1 \right\}^{n_{p}}$ and $r^{a} , d^{a} \in \left\{ 0 , 1 \right\}^{n_{a}}$ where $n_{p}$ and $n_{a}$ are the numbers of primary and ancillary qubits respectively.
    \label{alg line:real IQP resource honest server - primary and ancila random key generation}
    \STATE The Client generates the states $\ket{p_{j}} = Z^{r^{p}_{j}} S^{d^{p}_{j}} \ket{+} $ and $\ket{a_{i}}  =  Z^{r^{a}_{i}} S^{d^{a}_{i}} \ket{+}$ for $j \in \left\{ 1 , \dots, n_{p} \right\}$ and $i \in \left\{ 1 , \dots, n_{a} \right\}$
    \label{alg line:real IQP resource honest server - primary and ancillary state generation}
    \STATE Client creates $d^b \in \left\{0,1\right\}^{n_b}$ in the following way: For $i=1,\dots,n_a$ and $j=1,\dots,n_p$, if $\widetilde{\mathbf{Q}}_{ij}=-1$ and $\mathbf{Q}_{ij}=0$, then $d^b_k=0$ else if $\widetilde{\mathbf{Q}}_{ij}=-1$ and $\mathbf{Q}_{ij}=1$ then $d^b_k=1$. He keeps track of the relation between $k$ and $(i,j)$ via the surjective function $g: \mathbb{Z}_{n_a \times n_p} \rightarrow \mathbb{Z}_{n_b}$.  
   
    \STATE The Client generates $r^{b} \in \left\{ 0 , 1 \right\}^{n_{b}}$ at random and produces the states $\ket{b_{k}} = Y^{r^{b}_{k}} \left( \sqrt{Y} \right)^{d^{b}_{k}} \ket{0} $ for $k \in \left\{ 1 , \dots, n_{b} \right\}$
    \label{alg line:real IQP resource honest server - bridge and break state generation}
    \STATE State $\rho$ comprising of all of the Client's produced states is sent to the Server.
    \STATE The Server implements $E_{\widetilde{\mathbf{Q}}}$.
    \STATE The Server measures qubits $b_{1} , \dots, b_{n_{b}}$ in the $Y$-basis $\left\{ \ket{+^{Y}} , \ket{-^{Y}} \right\}$ and sends the outcome $s^b \in \left\{ 0 , 1 \right\}^{n_{b}}$ to the Client.
    \STATE The Client calculates $\Pi^{z} , \Pi^{s} \in \left\{ 0 , 1 \right\}^{n_{p}}$ and $A^{z} , A^{s} \in \left\{ 0 , 1 \right\}^{n_{a}}$ using equations \eqref{equ:primary Z correction term}, \eqref{equ:primary S correction term}, \eqref{equ:ancila Z correction term} and \eqref{equ:ancila S correction term}. 
    \begin{align}
      \label{equ:primary Z correction term}
      \Pi^{z}_{j} &= \sum_{i,k:g(i,j)=k} r_k^b \left( 1 - d^{b}_k \right) - r^{p}_{j} \\
      \label{equ:primary S correction term}
      \Pi^{s}_{j} &= \sum_{i,k:g(i,j)=k} (-1)^{s^{b}_k+r^{b}_k} d^{b}_k  - d^{p}_{j}\\
      \label{equ:ancila Z correction term}
      A^{z}_{i} &= \sum_{j,k:g(i,j)=k} r_k^{b} \left( 1 - d^{b}_k \right) - r^{a}_{i}\\
      \label{equ:ancila S correction term}
      A^{s}_{i} &= \sum_{j,k:g(i,j)=k}(-1)^{s^{b}_k+r^{b}_k}  d^{b}_k  - d^{a}_i
    \end{align}
    \STATE The Client sends $A \in \left\{0,1,2,3\right\}^{n_a}$ and $\Pi \in \left\{0,1,2,3\right\}^{n_p}$ for the ancillary and primary qubits respectively, where $A_{i} = A^{s}_{i} + 2 A^{z}_{i} \pmod 4$ and $\Pi_{j} = \Pi^{s}_{j} + 2 \Pi^{z}_{j} \pmod 4$.
    \STATE The Server measures the respective qubits in the basis below for the ancillary and primary qubits respectively. 
    \begin{equation}
        \label{equ:primary and ancillary measurement basis}
        S^{- A_{i} } \left\{ \ket{0_\theta} , \ket{1_\theta} \right\}  \text{ and } S^{- \Pi_{j} } \left\{ \ket{+} , \ket{-} \right\}
    \end{equation}
    The measurement outcomes $s^{p} \in \left\{ 0 , 1 \right\}^{n_{p}}$ and $s^{a} \in \left\{ 0 , 1 \right\}^{n_{a}}$ are sent to the Client.
    \STATE The Client generates and outputs $\widetilde{x} \in \left\{ 0 , 1 \right\}^{n_{p}}$ as follows. 
    \begin{equation}
     \label{equ:IQP final outcome calculation}
      \widetilde{x}_{j} = s^{p}_{j} + \sum_{i:\mathbf{Q}_{ij} = 1}  s^{a}_{i} \pmod 2
    \end{equation}
 \end{algorithmic}
 
\end{algorithm}

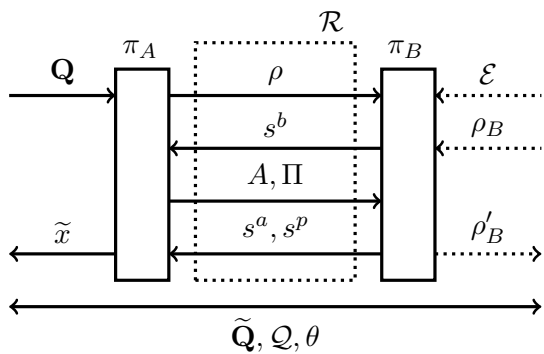
\begin{figure}
  \centering
  \begin{tikzpicture}[scale = 0.7]
    \draw[very thick] (0,0.5) rectangle (1,4.5) node[anchor = south east] {$\pi_{A}$};
    
    \draw[very thick, ->] (-2,4) -- (0,4) node[anchor = south, pos = 0.5] {$\mathbf{Q}$};
    \draw[very thick, <-] (-2,1) -- (0,1) node[anchor = south, pos = 0.5] {$\widetilde{x}$};
    
    \draw[very thick, ->] (1,4) -- (5,4) node[anchor = south, pos = 0.5] {$\rho$};
    \draw[very thick, <-] (1,3) -- (5,3) node[anchor = south, pos = 0.5] {$s^b$};
    \draw[very thick, ->] (1,2) -- (5,2) node[anchor = south, pos = 0.5] {$A,\Pi$};
    \draw[very thick, <-] (1,1) -- (5,1) node[anchor = south, pos = 0.5] {$s^a,s^p$};
    
    \draw[very thick] (5,0.5) rectangle (6,4.5) node[anchor = south east] {$\pi_{B}$};
    
    \draw[very thick, dotted, <-] (6,4) -- (8,4) node[anchor = south, pos = 0.5] {$\mathcal{E}$};
    \draw[very thick, dotted, <-] (6,3) -- (8,3) node[anchor = south, pos = 0.5] {$\rho_{B}$};
    \draw[very thick, dotted, ->] (6,1) -- (8,1) node[anchor = south, pos = 0.5] {$\rho_{B}'$};
    
    \draw[very thick, <->] (-2,0) -- (8,0) node[anchor = north, pos = 0.5] {$\widetilde{\mathbf{Q}} , \mathcal{Q}, \theta$};
    
    \draw[very thick, dotted] (1.5, 0.5) rectangle (4.5,5) node[anchor = south east] {$\mathcal{R}$};
  \end{tikzpicture}
  \caption{The real communication protocol of Algorithm \ref{alg:real IQP resource honest server}.}
  \label{fig:real IQP computation resource}
\end{figure}

During the execution of the protocol of Algorithm \ref{alg:real IQP resource honest server}, the Server sends two classical bit strings to the Client that correspond to the measurement outcomes of the sent qubits. If the Server wants to deviate from the protocol, he will again use some quantum map $\mathcal{E}$ on the information received so far together with the state $\rho_B$ he has in his own register. At the final step of the protocol the Server may output some quantum state $\rho_{B}'$.  

To prove composable security of the proposed protocol we drop the notion of a malicious Server for that of a global distinguisher that has a view of all inputs and outputs of the relevant resources. To recreate the view of a malicious Server, we develop a simulator $\sigma$ interfacing between the ideal resource $\mathcal{S}$ of Figure \ref{fig:ideal resource} and the distinguisher in such a way that the latter cannot tell the difference between an interaction with the ideal resource and the real protocol. We employ the Abstract Cryptography framework introduced in \cite{Abstract cryptography, Cryptographic security of quantum key distribution} and teleportation techniques inspired by \cite{Composable security of delegated quantum computation} to prove security in the case of a malicious Server. We will prove that:
\begin{equation}
    \pi_A\mathcal{R}\equiv \mathcal{S}\sigma
\end{equation}
where $\mathcal{R}$ is the communication channel (quantum and classical) used by the Client and the Server in the protocol.

\begin{theorem}
    \label{thm:security proof}
    The protocol described by Algorithm \ref{alg:real IQP resource honest server} is information theoretically secure against a dishonest Server. 
\end{theorem}

For the sake of brevity, we give only an intuitive proof here and leave a thorough proof to Appendix \ref{thm:security of real resource appendix}.

 \begin{proof}
    The proof consists of a pattern of transformations of the real protocol of Algorithm \ref{alg:real IQP resource honest server}, into the ideal resource plus simulator setting of Algorithm \ref{alg:real IQP resource honest server with simulator}. These transformations leave the computation unchanged, therefore ensuring the indistinguishability of the two settings and so the security of the protocol. As the computation itself is not changed by the transformations we also ensure that we are still sampling from the original IQP distribution, providing evidence for the correctness of Algorithm \ref{alg:real IQP resource honest server with simulator}.

\begin{algorithm}
  \caption{Blind distributed IQP computation with simulator}
  \label{alg:real IQP resource honest server with simulator}
  \textbf{Public:} $\widetilde{\mathbf{Q}} , \mathcal{Q}$, $\theta$
  
  \textbf{Client input:} $\mathbf{Q}$ 
  
  \textbf{Client output:} $\widetilde{x}$ 
  
  \textbf{The simulator}
  \begin{algorithmic}[1]
  \STATE Generates $n_{p}+n_a+n_b$ EPR pairs and sends half of each to the ideal resource and the other half to the distinguisher.
  \STATE Receives the bitstring  $s_{b} \in \left\{ 0 , 1 \right\}^{n_{b}}$ from the distinguisher and forwards it to the ideal resource.
  \STATE Randomly generates $\Pi \in \left\{ 0 , 1 , 2 , 3 \right\}^{n_{p}}$ and $A \in \left[ 0 , 1 , 2 , 3 \right]^{n_{a}}$ and sends them to the ideal resource and distinguisher.
  \STATE Receives the bitstrings  $s^{p} \in \left\{ 0 , 1 \right\}^{n_{p}}$ and $s^{a} \in \left\{ 0 , 1 \right\}^{n_{a}}$ from the distinguisher and forwards them to the ideal resource.
  \end{algorithmic}\vspace{0.1in}
  
  \textbf{The ideal resource}
  \begin{algorithmic}[1]
  \STATE Calculates $d^b \in \left\{0,1\right\}^{n_b}$ in the following way: For $i=1,\dots,n_a$ and $j=1,\dots,n_p$, if $\widetilde{\mathbf{Q}}_{ij}=-1$ and $\mathbf{Q}_{ij}=0$, then $d^b_k=0$ else if $\widetilde{\mathbf{Q}}_{ij}=-1$ and $\mathbf{Q}_{ij}=1$ then $d^b_k=1$. Keep track of the relation between $k$ and $(i,j)$ via the surjective function $g: \mathbb{Z}_{n_a \times n_p} \rightarrow \mathbb{Z}_{n_b}$. 
  \STATE Measures the corresponding half EPR pairs in the bases $\sqrt{Y}^{d^{b}_{k}} \left\{ \ket{0} , \ket{1} \right\}$ getting outcomes $r^{b}_k$, for $k=1,\dots,n_b$.
  \STATE Calculates $d^{p} \in \left\{ 0 , 1 , 2 , 3 \right\}^{n_{p}}$ and $d^{a} \in \left\{ 0 , 1 , 2 , 3 \right\}^{n_{a}}$ using equations \eqref{equ:primary measurement term} and \eqref{equ:ancila measurement term} respectively. 
    \begin{align}
        \label{equ:primary measurement term}
        d^{p}_{j} &= \sum_{i,k:g(i,j)=k}  (-1)^{s^{b}_k+r^{b}_k}d^{b}_k + 2 \sum_{i,k:g(i,j)=k} r_k^{b} \left( 1 - d^{b}_k \right)  - \Pi_{j} \\
        \label{equ:ancila measurement term}
        d^{a}_{i} &= \sum_{j,k:g(i,j)=k} (-1)^{s^{b}_k+r^{b}_k}d^{b}_k + 2 \sum_{j,k:g(i,j)=k} r_k^{b} \left( 1 - d^{b}_k \right) - A_{i}
    \end{align}
  \STATE Measures the remaining half of the EPR pairs corresponding to the ancillary and primary qubits in the bases $S^{d^{a}_{i}} \left\{ \ket{+} , \ket{-} \right\}$ and $S^{d^{p}_{j}} \left\{ \ket{+} , \ket{-} \right\}$, getting outcomes $r^{a}_{i}$ and $r^{p}_{j}$ for $i=1,\dots,n_a$ and $j=1,\dots,n_p$ respectively.
 \STATE Generates and outputs $\widetilde{x} \in \left\{ 0 , 1 \right\}^{n_{p}}$ using equation \eqref{equ:IQP final outcome calculation with corrections}.
    \begin{equation}
        \label{equ:IQP final outcome calculation with corrections}
        \widetilde{x}_{j} = \left( s^{p}_{j} + r^{p}_{j} \right) + \sum_{i:\mathbf{Q}_{ij} = 1}  \left( s^{a}_{i} + r^{a}_{i} \right)
    \end{equation}
  \end{algorithmic}
 
\end{algorithm}
  
    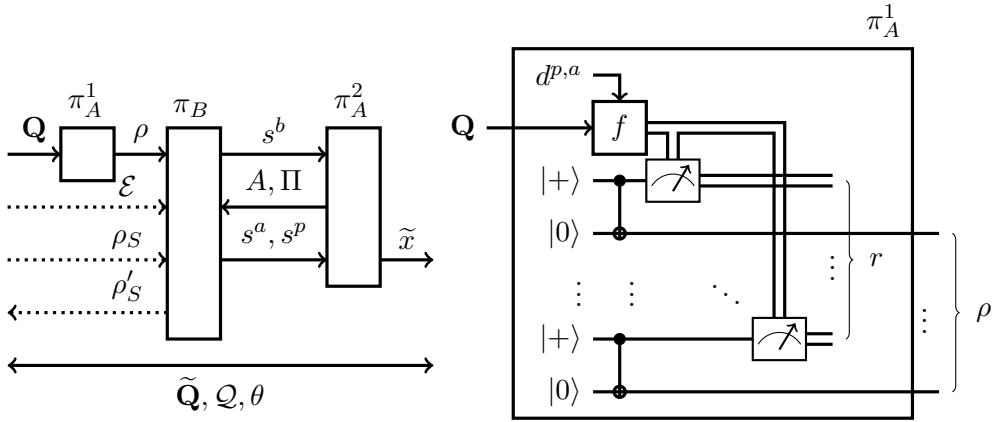
\begin{figure}
        \centering
        \begin{tikzpicture}[scale = 0.7]
            \draw[very thick, ->] (-8,5.5) -- (-7,5.5) node[anchor = south, pos = 0.5] {$\mathbf{Q}$};
            \draw[very thick] (-7,5) rectangle (-6,6) node[anchor = south east] {$\pi_{A}^{1}$};
            \draw[very thick, ->] (-6,5.5) -- (-5,5.5) node[anchor = south, pos = 0.5] {$\rho$};
    
            \draw[very thick, dotted, <-] (-5,4.5) -- (-8,4.5) node[anchor = south, pos = 0.25] {$\mathcal{E}$}; 
            \draw[very thick, dotted, <-] (-5,3.5) -- (-8,3.5) node[anchor = south, pos = 0.25] {$\rho_{S}$};
            \draw[very thick, dotted, ->] (-5,2.5) -- (-8,2.5) node[anchor = south, pos = 0.25] {$\rho_{S}'$};
    
            \draw[very thick] (-5,2) rectangle (-4,6) node[anchor = south east] {$\pi_B$};
    
            \draw[very thick, ->] (-4,5.5) -- (-2,5.5) node[anchor = south, pos = 0.5] {$s^b$};
            \draw[very thick, <-] (-4,4.5) -- (-2,4.5) node[anchor = south, pos = 0.5] {$A, \Pi$};
            \draw[very thick, ->] (-4,3.5) -- (-2,3.5) node[anchor = south, pos = 0.5] {$s^a,s^p$};
    
            \draw[very thick] (-2,3) rectangle (-1,6) node[anchor = south east] {$\pi_{A}^{2}$};
      
            \draw[very thick, ->] (-1,3.5) -- (0,3.5) node[anchor = south, pos = 0.5] {$\widetilde{x}$};
    
            \draw[very thick, <->] (-8,1.5) -- (0,1.5) node[anchor = north, pos = 0.5] {$\widetilde{\mathbf{Q}} , \mathcal{Q},\theta$};
            
            \draw[very thick] (1.5,0.5) rectangle (9,7.5) node[anchor = south east] {$\pi_{A}^{1}$};
    
            \draw[very thick, ->] (1,6) node[anchor = east] {$\mathbf{Q}$} -- (3,6);
            \draw[very thick] (3,5.5) rectangle (4,6.5) node[align = center, pos = 0.5] {$f$};
            \draw[very thick] (4,6.1) -- (6.6,6.1) -- (6.6,2.4);
            \draw[very thick] (4,5.9) -- (4.4,5.9) -- (4.4,5.4);
            \draw[very thick] (6.4,2.4) -- (6.4,5.9) -- (4.6,5.9) -- (4.6,5.4);
    
            \draw[very thick] (3,5) node[anchor = east] {$\ket{+}$} -- (4,5);
            \draw[very thick] (3,4) node[anchor = east] {$\ket{0}$} -- (9.5,4);
    
            \node[anchor = east] at (3,3) {$\vdots$};
            \node[anchor = east] at (4,3) {$\vdots$};
            \node[anchor = east] at (9.5,2.5) {$\vdots$};
    
            \draw[very thick] (3,2) node[anchor = east] {$\ket{+}$} -- (6,2);
            \draw[very thick] (3,1) node[anchor = east] {$\ket{0}$} -- (9.5,1);
    
            \filldraw (3.5,5) circle (3pt);
            \draw[very thick] (3.5,4) circle (3pt);
            \draw[very thick] (3.5,5) -- (3.5,3.9);
    
            \filldraw (3.5,2) circle (3pt);
            \draw[very thick] (3.5,1) circle (3pt);
            \draw[very thick] (3.5,2) -- (3.5,0.9);
    
            \draw[thick] (4,4.6) rectangle (5,5.4);
            \draw (4.1,4.9) .. controls (4.3,5.2) and (4.7,5.2) .. (4.9,4.9);
            \draw[thick, ->] (4.5, 4.8) -- (4.8, 5.3);
    
            \node at (5.5,3) {$\ddots$};
    
            \draw[thick] (6,1.6) rectangle (7,2.4);
            \draw (6.1,1.9) .. controls (6.3,2.2) and (6.7,2.2) .. (6.9,1.9);
            \draw[thick, ->] (6.5, 1.8) -- (6.8, 2.3);
    
            \draw[decoration={brace,mirror,raise=5pt},decorate] (9.5,1) -- (9.5,4) node[right = 10pt, pos = 0.5] {$\rho$};
    
            \draw[very thick] (5,5.1) -- (7.5,5.1);
            \draw[very thick] (5,4.9) -- (7.5,4.9);
    
            \node at (7.5,3.5) {$\vdots$};
    
            \draw[very thick] (7,2.1) -- (7.5,2.1);
            \draw[very thick] (7,1.9) -- (7.5,1.9);
    
            \draw[decoration={brace,mirror,raise=5pt},decorate] (7.5,2) -- (7.5,5) node[right = 10pt, pos = 0.5] {$r$};
    
            \draw[very thick, ->] (3,7) node[anchor = east] {$d^{p,a}$} -- (3.5,7) -- (3.5,6.5);
        \end{tikzpicture}
        \caption{The real protocol with the state generation phase of the protocol, $\pi_{A}^{1}$ isolated (left) and further analysed (right) using an equivalent protocol based on teleportation, where $f$ represents the measurement angle calculation on one half of the EPR pairs (see Algorithm \ref{alg:real IQP resource honest server with added teleportation} in Appendix for details).}
        \label{fig:real ideal resource rearranged}
    \end{figure}

  Line \ref{alg line:real IQP resource honest server - primary and ancillary state generation} of Algorithm \ref{alg:real IQP resource honest server} generates at random one of the four states $\ket{+}$, $\ket{+^{Y}}$, $\ket{-}$ and $\ket{-^{Y}}$. The same effect is achieved by measuring an EPR pair with equal probability in one of the bases $\left\{ \ket{+} , \ket{-} \right\}$ and $\left\{ \ket{+^{Y}} , \ket{-^{Y}} \right\}$. The application of the $\left( \sqrt{Y} \right)^{d^{b}_{k}}$ operation in line \ref{alg line:real IQP resource honest server - bridge and break state generation} of Algorithm \ref{alg:real IQP resource honest server} decides, according to the graph to be created, if the bridge and break qubit will be drawn from the set $\left\{ \ket{+} , \ket{-} \right\}$ or $\left\{ \ket{0} , \ket{1} \right\}$. Using the same information to choose between the measurement bases $\left\{ \ket{+} , \ket{-} \right\}$ and $\left\{ \ket{0} , \ket{1} \right\}$ on one half of an EPR pair has the same effect. The random rotation $Y^{r^{b}_{k}}$ then has the same effect of the randomness that is intrinsic to the EPR pair measurement. This may be visualised in Figure \ref{fig:real ideal resource rearranged} which presents a simple rearrangement of the Real Resource of Figure \ref{fig:real IQP computation resource} in order to isolate the state generation phase $\pi_{A}^{1}$ and to examine an equivalent circuit based on teleportation.

 The next transformation is to delay the first measurement of the EPR pairs as implied in Figure \ref{fig:pre generated randomness resource}. Since information about the measurement outcome  $r$ is not yet available to define $\Pi$ and $A$, the Client chooses random $\Pi$ and $A$ which will then corrected for by using these values to compute the measurement bases for the Client's half of the primary and ancillary EPR pairs. 
 
 Finally, Figure \ref{fig:rearranged into simulator setting} simply involves a rearrangement of the players in Figure \ref{fig:pre generated randomness resource} to match those in the simulator/distinguisher setting. The formal description of the protocol displayed by Figure \ref{fig:rearranged into simulator setting} is seen in Algorithm \ref{alg:real IQP resource honest server with simulator}.
  
    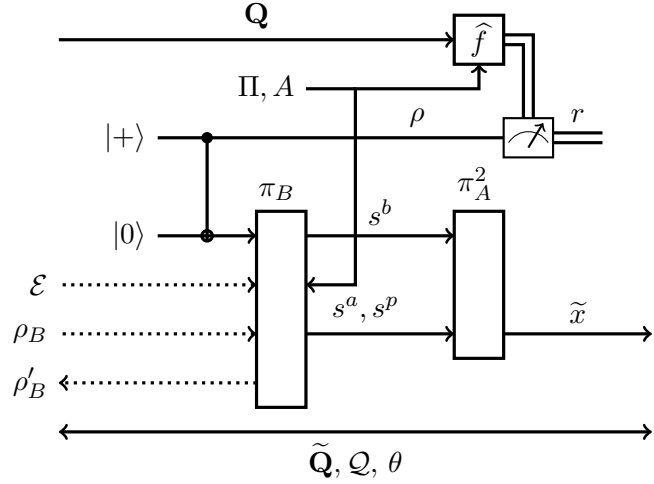
\begin{figure}
        \centering
     \begin{tikzpicture}[scale = 0.65]
            \draw[very thick, ->] (0,8) -- (8,8) node[anchor = south, pos = 0.5] {$\mathbf{Q}$};
            \draw[very thick] (8,7.5) rectangle (9,8.5) node[align = center, pos = 0.5] {$\widehat{f}$};
    
            \draw[very thick] (2,6) node[anchor = east] {$\ket{+}$} -- (9,6) node[pos = 0.75, anchor = south] {$\rho$};
            \draw[very thick, ->] (2,4) node[anchor = east] {$\ket{0}$} -- (4,4);
    
            \filldraw (3,6) circle (3pt);
            \draw[very thick] (3,4) circle (3pt);
            \draw[very thick] (3,6) -- (3,3.9);
    
            \draw[very thick, <->] (5,3) -- (6,3) -- (6,7) -- (8.5,7) -- (8.5,7.5);
            \draw[very thick] (6,7) -- (5,7) node[anchor = east] {$\Pi, A$};
    
            \draw[thick] (9,5.6) rectangle (10,6.4);
            \draw (9.1,5.9) .. controls (9.3,6.2) and (9.7,6.2) .. (9.9,5.9);
            \draw[thick, ->] (9.5, 5.8) -- (9.8, 6.3);
    
            \draw[very thick] (10,6.1) -- (11,6.1) node[anchor = south, pos = 0.5] {$r$};
            \draw[very thick] (10,5.9) -- (11,5.9);
   
            \draw[very thick] (9,8.1) -- (9.6,8.1) -- (9.6,6.4);
            \draw[very thick] (9,7.9) -- (9.4,7.9) -- (9.4,6.4);
    
            \draw[very thick, dotted, <-] (4,3) -- (0,3) node[anchor = east] {$\mathcal{E}$};
            \draw[very thick, dotted, <-] (4,2) -- (0,2) node[anchor = east] {$\rho_{B}$};
            \draw[very thick, dotted, ->] (4,1) -- (0,1) node[anchor = east] {$\rho_{B}'$};
    
            \draw[very thick] (4,0.5) rectangle (5,4.5) node[anchor = south east] {$\pi_B$};
    
            \draw[very thick, ->] (5,4) -- (8,4) node[anchor = south, pos = 0.5] {$s^{b}$};
            \draw[very thick, ->] (5,2) -- (8,2) node[anchor = south, pos = 0.4] {$s^{a} , s^{p}$};
    
            \draw[very thick] (8,1.5) rectangle (9,4.5) node[anchor = south east] {$\pi_{A}^{2}$};
    
            \draw[very thick, ->] (9,2) -- (12,2) node[anchor = south, pos = 0.5] {$\widetilde{x}$};
    
    
            \draw[very thick, <->] (0,0) -- (12,0) node[anchor = north, pos = 0.5] {$\widetilde{\mathbf{Q}} , \mathcal{Q}$, $\theta$};

        \end{tikzpicture}   
        \caption{The real protocol with only one input qubit for simplicity, where the Client sends random measurement instructions $A,\Pi$ to the Server and delays the teleportation measurement until after the Server has sent the measurement outcomes $s= \left\{ s^a,s^b,s^p \right\}$. $r = \left\{ r^{p} , r^{a} , r^{b} \right\}$. Here $\widehat{f}$ represents the process of calculating measurement angles to be performed on one half of the EPR pair from Eqs. \eqref{equ:primary measurement term} and \eqref{equ:ancila measurement term} (for details see Algorithm \ref{alg:real IQP resource honest server with added teleportation, rearrangement and pre-made randomness} in the Appendix). }
        \label{fig:pre generated randomness resource}
    \end{figure}

  \begin{figure}
    \centering
    \begin{tikzpicture}[scale = 0.65]
      \draw[very thick, ->] (1,7) -- (3,7) node[anchor = south, pos = 0.5] {$\mathbf{Q}$};
      \draw[very thick] (3,6.5) rectangle (4,7.5) node[align = center, pos = 0.5] {$\widehat{f}$};
    
      \draw[very thick, ->] (9,5) node[anchor = west] {$\ket{+}$} -- (5.5,5) -- (5.5,5.6);
      \draw[very thick, ->] (8,6) node[anchor = east] {$\ket{0}$} -- (11,6) node[anchor = south, pos = 0.825] {$\rho$};
      \filldraw (8.5,6) circle (3pt);
      \draw[very thick] (8.5,5) circle (3pt);
      \draw[very thick] (8.5,6) -- (8.5,4.9);
      \draw[very thick, dotted] (7,1) rectangle (10,8) node[anchor = south east] {$\sigma$};
     
      \draw[thick] (5,5.6) rectangle (6,6.4);
      \draw (5.1,5.9) .. controls (5.3,6.2) and (5.7,6.2) .. (5.9,5.9);
      \draw[thick, ->] (5.5, 5.8) -- (5.8,6.3);
     
      \draw[very thick] (4,7.1) -- (5.6,7.1) -- (5.6,6.4);
      \draw[very thick] (4,6.9) -- (5.4,6.9) -- (5.4,6.4);
     
      \draw[very thick, dotted, <-] (12,7) -- (13,7) node[anchor = south, pos = 0.5] {$\mathcal{E}$};
      \draw[very thick, dotted, <-] (12,6) -- (13,6) node[anchor = south, pos = 0.5] {$\rho_{B}$};
      \draw[very thick, dotted, ->] (12,2) -- (13,2) node[anchor = south, pos = 0.5] {$\rho_{B}'$};
     
      \draw[very thick] (11,1) rectangle (12,8) node[anchor = south east] {$\pi_B$};
     
      \draw[very thick, ->] (11,4) -- (4,4) node[anchor = south, pos = 0.35] {$s^{b}$};
      \draw[very thick, <->] (11,3) -- (4.5,3) node[anchor = south, pos = 0.4] {$\Pi , A$} -- (4.5,8) -- (3.5,8) -- (3.5,7.5);
      \draw[very thick, ->] (11,2) -- (4,2) node[anchor = south, pos = 0.35] {$s^{a} , s^{p}$};
     
      \draw[very thick] (3,1.5) rectangle (4,4.5) node[anchor = south east] {$\pi_{A}^{2}$};
     
      \draw[very thick, ->] (3,2) -- (1,2) node[anchor = south, pos = 0.5] {$\widetilde{x}$};
     
      \draw[very thick, dotted] (2.5,1) rectangle (6.5,8.5) node[anchor = south east] {$\mathcal{S}$};
     
      \draw[very thick] (5,6.1) -- (4,6.1);
      \draw[very thick] (5,5.9) -- (4,5.9);
      \node[anchor = east] at (4,6) {$r$};
     
      \draw[very thick, <->] (1,0) -- (13,0) node[anchor = north, pos = 0.5] {$\widetilde{\mathbf{Q}} , \mathcal{Q}$, $\theta$};
    \end{tikzpicture}
    \caption{The ideal resource $\mathcal{S}$ and the simulator $\sigma$ for the malicious Server, shown with only one input qubit for simplicity. The simulator has no access to the private information $\mathbf{Q}$ and any time. A global distinguisher cannot tell the difference between this setting and the real protocol.}
    \label{fig:rearranged into simulator setting}
  \end{figure}
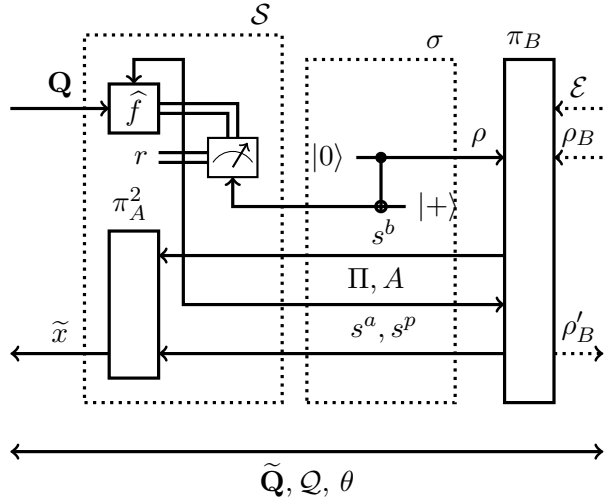
  
\end{proof}

We can now be sure that our communication protocol is indistinguishable from an ideal resource (defined in Figure \ref{fig:ideal resource}) which performs an IQP computation without communicating any information to the Server which is not already public. This means that the communication protocol does not reveal any information about the computation to the Server. Furthermore, this is proven in a composable framework \cite{Abstract cryptography,Cryptographic security of quantum key distribution,Composable security of delegated quantum computation} and so can be used as part of future protocols as we will in section \ref{sec:hypothesis test}.


\section{The Hypothesis Test}
\label{sec:hypothesis test}

\subsection{Previous work}
\label{sec:SB hypothesis test}

We now have all the tools to form a test for a Server to run in order to prove to a Client that they are capable of solving classically non-simulatable problems. Specifically, we ask the Server to perform an IQP computation that we believe is classically hard, but whose solution can easily be checked by a classical Client.

The general idea of our \emph{Hypothesis Test}, building on the work of \cite{Temporally_Unstructured_Quantum_Computation}, is that there is some hidden structure in the program elements, $\mathbf{q}_i$, of the $X$-program that results in some structure in the distribution of the outputs, known only to the Client. The Client can use this structure to check the Server's reply. 
A Server possessing an IQP machine can reproduce this structure by implementing the $X$-program. A Server not in possession of an IQP machine cannot generate outputs obeying the same rules.

We summarise this discussion by three conditions that a hypothesis test using this method must meet.

\begin{enumerate}[label = \arabic{list}.\arabic*]
	\item \label{pt:hypothesis test IQP hard} The Client asks the Server to perform an IQP computation that is hard to classically simulate. 
	\item \label{pt:hypothesis test structure} The Client can check the solution to this computation because they know some secret structure that makes this checking processes efficient.
	\item \label{pt:hypothesis test hidden structure} The Server must be unable to uncover this structure in polynomial time.
\end{enumerate}
\stepcounter{list}

\noindent A particular `known structure' of the output which is used in \cite{Temporally_Unstructured_Quantum_Computation} to satisfy condition \ref{pt:hypothesis test structure} is its \emph{bias}. 

\begin{definition}
	\label{def:bias}
	If X is a random variable taking values in $\left\{ 0 , 1 \right\}^{n_p}$ and $\mathbf{s} \in \left\{ 0 , 1 \right\}^{n_p}$ then the bias of $X$ in the direction $\mathbf{s}$ is $\mathbb{P} \left( X \cdot \mathbf{s}^T = 0 \right)$ where the product is performed modulo 2. Hence, the bias of a distribution in the direction $\mathbf{s}$ is the probability of a sample from the distribution being orthogonal to $\mathbf{s}$.
\end{definition}

To calculate the bias of $X$ in direction $\mathbf{s} \in \left\{ 0 , 1 \right\}^{n}$, we form the linear code $\mathcal{C}_{\mathbf{s}}$ by selecting all rows, $\mathbf{q}_{i}$ of the X-program, $\left(\mathbf{Q}, \theta \right) \in \left\{ 0 , 1 \right\}^{n_{a} \times n_{p}} \times \left[ 0 , 2 \pi \right]$, such that $\mathbf{q}_{i} \cdot \bf{s}^{T} = 1$ and forming, from them, a new matrix, $\mathbf{Q}_{\mathbf{s}}$, which is the generator matrix of $\mathcal{C}_{\bf{s}}$. Defining $n_{\bf{s}}$ to be the number of rows of $\mathbf{Q}_{\bf{s}}$ allows us to understand the following expression. The derivation can be found in \cite{Temporally_Unstructured_Quantum_Computation}.
\begin{equation}
	\label{equ:bias prediction}
	\mathbb{P} \left( X \cdot \mathbf{s}^{T} = 0 \right) = \mathbb{E}_{ \mathbf{c} \sim \mathcal{C}_{\mathbf{s}}} \left[ \cos^{2} \left( \theta \left( n_s - 2 \cdot \# \mathbf{c} \right) \right) \right]
\end{equation}

We find that the bias of an X-program in the direction $\bf{s}$ depends only on $\theta$ and the linear code defined by the generator matrix $\mathbf{Q}_{\mathbf{s}}$. One can now imagine a hypothesis test derived from these facts. Although the X-program that will be implemented, needs to be made public, the direction $\mathbf{s}$ which will be used for checking, will be kept secret. This gives a Client, with the computational power to calculate the quantity of expression \eqref{equ:bias prediction}, the necessary information to compute the bias, but does not afford the Server the same privilege. 

What we want to show is that the only way for the Server to produce an output with the correct bias is to use an IQP machine. If the Server could somehow uncover $\mathbf{s}$ then they could calculate the value of expression \eqref{equ:bias prediction} and return vectors to the Client which are orthogonal to $\mathbf{s}$ with the correct probability. We specialise the conditions mentioned at the beginning of this section to this particular method. 


\begin{enumerate}[label = \arabic{list}.\arabic*]
	\item \label{pt:bias hypothesis test IQP hard} The X-Program sent to a Server represents an IQP computation that is hard to classically simulate.
	\item \label{pt:bias hypothesis test structure} It must be possible for a Client, having knowledge of a secret $\mathbf{s}$ and the X-program, to calculate the quantity of expression \eqref{equ:bias prediction}.
	\item \label{pt:bias hypothesis test hidden structure} The knowledge of the Server must be insufficient to learn the value of $\bf{s}$.
\end{enumerate}
\stepcounter{list}


In \cite{Temporally_Unstructured_Quantum_Computation} the authors develop a protocol for building an $X$-program and a vector $\mathbf{s}$ performing this type of hypothesis test. The code $\mathcal{C}_{\mathbf{s}}$ used to build the $X$-program is a quadratic residue code with $\theta = \frac{\pi}{8}$. Condition \ref{pt:bias hypothesis test IQP hard} is conjectured, by \cite{Temporally_Unstructured_Quantum_Computation}, to be satisfied by these X-programs. This conjecture is supported by giving a classical simulation that is believed to be optimal and achieves maximum bias value $0.75$; different from that expected from an IQP machine. A hypothesis test with X-programs, such as the random circuits of \cite{Classical Simulation of Commuting Quantum Computations Implies Collapse of the Polynomial Hierarchy}, for which connections to an implausible collapse in the polynomial hierarchy has been made, is an open problem. Condition \ref{pt:bias hypothesis test structure} is also satisfied by the construction in \cite{Temporally_Unstructured_Quantum_Computation}, by proving that the bias value, which is $\cos^{2}\left( \frac{\pi}{8} \right)$ for their choice of X-program and $\mathbf{s}$, can be calculated in polynomial time. 

The way in which condition \ref{pt:bias hypothesis test hidden structure} is addressed in \cite{Temporally_Unstructured_Quantum_Computation} relies on the fact that the right-hand side expression of Eq.\eqref{equ:bias prediction} is equal for all generator matrices in a \emph{matroid} \cite{Matroid Theory}.

\begin{definition}
	A $i$-point binary \emph{matroid} is an equivalence class of matrices with $i$ rows, defined over $\mathbb{F}_2$. Two matrices, $\mathbf{M}_1$ and $\mathbf{M}_2$, are said to be equivalent if, for some permutation matrix $\mathbf{R}$, the column echelon reduced form of $\mathbf{M}_1$ is the same as the column echelon reduced form of $\mathbf{R} \cdot \mathbf{M}_2$ (In the case where the column dimensions do not match, we define equivalence  by deleting columns containing only $0$s after column echelon reduction).
	
\end{definition}




In order to move to a new matrix within the same matroid, consider the right-multiplication with matrix $\mathbf{A}$ on $\mathbf{Q}$. Notice that $\mathbf{q}_i \mathbf{s}^{T} = \left( \mathbf{q}_i \mathbf{A} \right) \left( \mathbf{A}^{-1} \mathbf{s}^{T} \right)$. Rows which were originally non-orthogonal to $\mathbf{s}$ are now non-orthogonal to $\mathbf{A}^{-1} \mathbf{s}^{T}$, hence we can locate $\mathbf{Q}_{\mathbf{s}}$ in $\mathbf{Q}$ by using $\mathbf{A}^{-1} \mathbf{s}^{T}$.

A way to hide $\mathbf{s}$ is therefore to randomise it with such an operation $\mathbf{A}$. We now understand what to do to the $X$-program we are considering, so that the value of the bias does not change. To increase the hiding of $\mathbf{s}$, the matrix might also include additional rows orthogonal to $\mathbf{s}$, which do not affect the value of the bias. The combination of matrix randomisation and the addition of new rows makes it hard, as conjectured in \cite{Temporally_Unstructured_Quantum_Computation}, up to some computational complexity assumptions, for the Server to recover $\mathbf{s}$ from the matrix that it receives. It is now simply a matter for the Server to implement the $X$-program and for the Client to check the bias of the output in the direction $\mathbf{s}$. This is the approach used by \cite{Temporally_Unstructured_Quantum_Computation} to address condition \ref{pt:bias hypothesis test hidden structure}.


%
%
	


\subsection{Our Protocol}
\label{sec:our protocol}
The main contribution of this work is to revisit condition \ref{pt:bias hypothesis test hidden structure}. By giving to the Client limited quantum capabilities, we remove the computational assumption of  \cite{Temporally_Unstructured_Quantum_Computation}, and therefore provide unconditional security against a powerful quantum Server. In Algorithm \ref{alg:Our hypothesis test protocol} we provide a hypothesis test that uses the blind delegated IQP computation resource of the previous section to verify quantum supremacy.

\begin{theorem}
    Algorithm \ref{alg:Our hypothesis test protocol} presents an information-theoretically secure solution to condition \ref{pt:bias hypothesis test hidden structure}.
\end{theorem}

\begin{proof}
    Let us begin by recalling that when $\mathcal{C}_{\mathbf{s}}$ in expression \eqref{equ:bias prediction} is the quadratic residue code space then we know that the value of that expression is $\cos^{2} \frac{\pi}{8}$.
    
    Notice that, in particular, the all one vector is in the quadratic residue code space. As such the matrix $\mathbf{Q}_{\mathbf{s}}$, introduced on line \ref{alg line: Our hypothesis test protocol - start with QR} of Algorithm \ref{alg:Our hypothesis test protocol},  which is the quadratic code generator matrix $\mathbf{Q}_{r}$ with a column of all ones appended to it also generates the quadratic residue code. 
    
    The vector $\mathbf{s} \in \left\{ 0 , 1 \right\}^{n_{p}}$ with all zero entries, with the exception of the last entry which is set to one, is non orthogonal to all rows of $\mathbf{Q}_{\mathbf{s}}$. Hence, adhering to the notation here and of Section \ref{sec:SB hypothesis test}, $\mathcal{C}_{\mathbf{s}}$ is the quadratic residue code and expression \eqref{equ:bias prediction} is equal to $\cos^{2} \frac{\pi}{8}$. The reader may refer to Figures \ref{fig:bipartite graph quadratic residue code} and \ref{fig:bipartite graph quadratic residue code extended by all one vector} for some intuition about the above matrices.
    
    $\mathbf{A}$, defined in line \ref{alg line: Our hypothesis test protocol - randomisation matrix} of Algorithm \ref{alg:Our hypothesis test protocol}, is the operation which adds columns of $\mathbf{Q}_{s}$, chosen when $\widehat{\mathbf{s}_{i}} = 1$, to the last column of $\mathbf{Q_{s}}$. We know that the resulting matrix, $\mathbf{Q} = \mathbf{Q_{s}} \mathbf{A}$, is also a generator matrix of the quadratic residue code as all the columns of $\mathbf{Q}_{s}$ are in the quadratic residue code space. We also know, from the discussion of Section \ref{sec:SB hypothesis test}, that all the rows of $\mathbf{Q}$ are non-orthogonal to $\mathbf{A}^{-1} \mathbf{s}^{T}$. As such $\mathcal{C}_{\mathbf{A}^{-1} \mathbf{s}^{T}}$, when $\mathbf{Q}$ is the $X$-program of concern, is the quadratic residue code space and hence the bias of the $X$-program $\mathbf{Q}$ in the direction $\mathbf{A}^{-1} \mathbf{s}^{T}$ is $\cos^{2} \frac{\pi}{8}$. This matrix may be visualised in Figure \ref{fig:bipartite graph quadratic residue code extended by all one vector randomised} and this fact is exploited in line \ref{alg line:Our hypothesis test protocol - test orthogonality} of Algorithm \ref{alg:Our hypothesis test protocol}.
    
    We know, however, that from any $\mathbf{Q}$ we can make the IQP extended graph $\widetilde{\mathbf{Q}}$, which is the matrix $\mathbf{Q}_{r}$ with a column of minus ones appended to the end. Observing Figure \ref{fig:bipartite graph quadratic residue code extended by all one vector IQP extended} may help to visualise this. We can now use the resource of Section \ref{sec:security proof} to perform a blind IQP computation.
    
    By using the blind IQP computation resource of Section \ref{sec:security proof} we have solved condition \ref{pt:bias hypothesis test hidden structure} but do so now with information theoretic security as opposed to the reliance on computational complexity assumptions used by \cite{Temporally_Unstructured_Quantum_Computation}. This is true because, as a product of using the resource of Section \ref{sec:security proof}, the Server learns only the distribution $\mathcal{Q}$ over the possible set of graphs $\mathbf{Q}$. By setting $\mathbf{Q} = \mathbf{Q_{s}} \mathbf{A}$,  Algorithm \ref{alg:Our hypothesis test protocol} develops a bijection mapping $\widehat{\mathbf{s}} \in \{0,1\}^{n_{p} - 1}$ to a unique matrix $\mathbf{Q} \in \{0,1\}^{n_{a} \times n_{p}}$. So $\mathcal{Q}$ is equivalent to the distribution from which $\widehat{\mathbf{s}}$ is drawn. In this case it is the uniform distribution over a set of size $2^{n_{p} - 1}$.
\end{proof}



\begin{figure}
	\centering
	\begin{tikzpicture}[scale = 0.7]
		\filldraw (1.5,0) circle (3pt) node[anchor = north] {$a_{1}$};
		\node at (3,0) {$\dots$};
		\filldraw (4.5,0) circle (3pt) node[anchor = north ] {$a_{n_{a}}$};
		
		\filldraw (0,2) circle (3pt) node[anchor = south ] {$p_{1}$};
		\draw[very thick, dotted] (0,2) -- (1.5,0);
		
		\filldraw (3,2) circle (3pt) node[anchor = south ] {$p_{2}$};
		\draw[very thick, dotted] (3,2) -- (4.5,0);
		
		\node at (4.5,2) {$\dots$};
		
		\filldraw (6,2) circle (3pt) node[anchor = south ] {$p_{n_{p} - 1}$};
		\draw[very thick, dotted] (6,2) -- (1.5,0);

		\node at (13,1) {$\mathbf{Q}_{r}$};
	\end{tikzpicture}
	\caption{Quadtatic residue code generator matrix, $\mathbf{Q}_{r}$, and the graph that it describes. Note that, to save space, this is only illustrative and that the connections in this image do not correspond to an actual quadratic residue code. This is implied by the dotted lines.}
	\label{fig:bipartite graph quadratic residue code}
\end{figure}
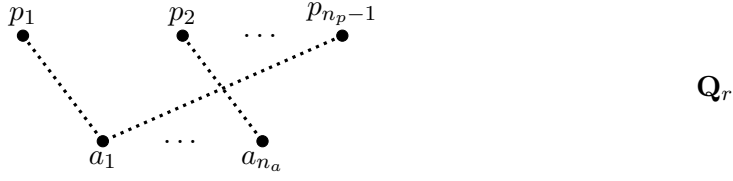

\begin{figure}
	\centering
	\begin{tikzpicture}[scale = 0.7]
	\filldraw (1.5,0) circle (3pt) node[anchor = north ] {$a_{1}$};
		\node at (3,0) {$\dots$};
		\filldraw (4.5,0) circle (3pt) node[anchor = north ] {$a_{n_{a}}$};
		
		\filldraw (0,2) circle (3pt) node[anchor = south ] {$p_{1}$};
		\draw[very thick, dotted] (0,2) -- (1.5,0);
		
		\filldraw (3,2) circle (3pt) node[anchor = south ] {$p_{2}$};
		\draw[very thick, dotted] (3,2) -- (4.5,0);
		
		\node at (4.5,2) {$\dots$};
		
		\filldraw (6,2) circle (3pt) node[anchor = south ] {$p_{n_{p} - 1}$};
		\draw[very thick, dotted] (6,2) -- (1.5,0);
		
		\filldraw (9,2) circle (3pt) node[anchor = south ] {$p_{n_{p}}$};
		
		\draw[very thick] (9,2) -- (1.5,0);
		\draw[very thick] (9,2) -- (4.5,0);
				
		\node at (13,1) {$\mathbf{Q}_{s} =   \left( 
		                                    \begin{array}{c|c}
                                                & 1\\
                                                \mathbf{Q}_{r} & \vdots\\
                                                & 1
                                            \end{array}
                                        \right)$};
	\end{tikzpicture}
	\caption{A matrix also generating the quadratic residue code space.}
	\label{fig:bipartite graph quadratic residue code extended by all one vector}
\end{figure}
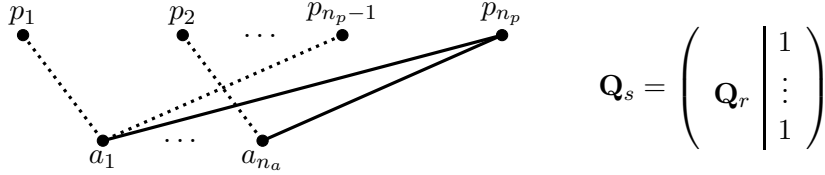

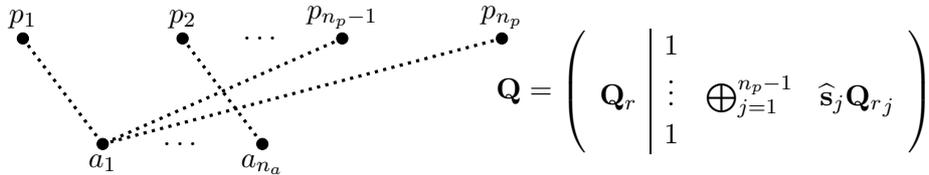
\begin{figure}
	\centering
	\begin{tikzpicture}[scale = 0.7]
		\filldraw (1.5,0) circle (3pt) node[anchor = north ] {$a_{1}$};
		\node at (3,0) {$\dots$};
		\filldraw (4.5,0) circle (3pt) node[anchor = north ] {$a_{n_{a}}$};
		
		\filldraw (0,2) circle (3pt) node[anchor = south ] {$p_{1}$};
		\draw[very thick, dotted] (0,2) -- (1.5,0);
		
		\filldraw (3,2) circle (3pt) node[anchor = south ] {$p_{2}$};
		\draw[very thick, dotted] (3,2) -- (4.5,0);
		
		\node at (4.5,2) {$\dots$};
		
		\filldraw (6,2) circle (3pt) node[anchor = south ] {$p_{n_{p} - 1}$};
		\draw[very thick, dotted] (6,2) -- (1.5,0);
		
		\filldraw (9,2) circle (3pt) node[anchor = south ] {$p_{n_{p}}$};
		
		\draw[very thick, dotted] (9,2) -- (1.5,0);
				
		\node at (13,1) {$\mathbf{Q} =   \left( 
		                                    \begin{array}{c|ccc}
                                                & 1\\
                                                \mathbf{Q}_{r} & \vdots & \bigoplus_{j = 1}^{n_{p} - 1} & \widehat{\mathbf{s}}_{j} {\mathbf{Q}_{r}}_{j} \\
                                                & 1
                                            \end{array}
                                        \right)$};
	\end{tikzpicture}
	\caption{A randomised version of Figure \ref{fig:bipartite graph quadratic residue code extended by all one vector}. Here ${\mathbf{Q}_{r}}_{j}$ is the $j^{\text{th}}$ column of $\mathbf{Q}_{r}$}
	\label{fig:bipartite graph quadratic residue code extended by all one vector randomised}
\end{figure}

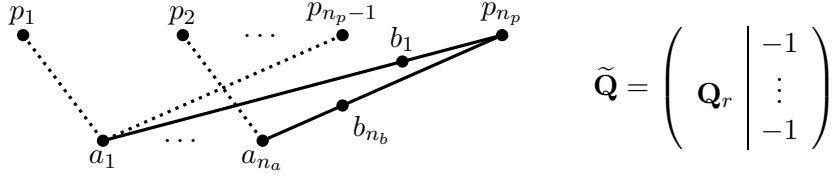
\begin{figure}
	\centering
	\begin{tikzpicture}[scale = 0.7]
		\filldraw (1.5,0) circle (3pt) node[anchor = north ] {$a_{1}$};
		\node at (3,0) {$\dots$};
		\filldraw (4.5,0) circle (3pt) node[anchor = north ] {$a_{n_{a}}$};
		
		\filldraw (0,2) circle (3pt) node[anchor = south ] {$p_{1}$};
		\draw[very thick, dotted] (0,2) -- (1.5,0);
		
		\filldraw (3,2) circle (3pt) node[anchor = south ] {$p_{2}$};
		\draw[very thick, dotted] (3,2) -- (4.5,0);
		
		\node at (4.5,2) {$\dots$};
		
		\filldraw (6,2) circle (3pt) node[anchor = south ] {$p_{n_{p} - 1}$};
		\draw[very thick, dotted] (6,2) -- (1.5,0);
		
		\filldraw (9,2) circle (3pt) node[anchor = south ] {$p_{n_{p}}$};
		
		\draw[very thick] (9,2) -- (1.5,0);
		\draw[very thick] (9,2) -- (4.5,0);
		
		\filldraw (1.5 + 45/8,6/4) circle (3pt) node[anchor = south] {$b_{1}$};
		\filldraw (4.5 + 4.5/3,2/3) circle (3pt) node[anchor = north west] {$b_{n_{b}}$};
		
		\node at (13,1) {$\widetilde{\mathbf{Q}} =   \left( 
		                                    \begin{array}{c|c}
                                                & -1\\
                                                \mathbf{Q}_{r} & \vdots\\
                                                & -1
                                            \end{array}
                                        \right)$};
	\end{tikzpicture}
	\caption{An IQP extended graph of all possible $\mathbf{Q}$ of Figure \ref{fig:bipartite graph quadratic residue code extended by all one vector randomised}}
	\label{fig:bipartite graph quadratic residue code extended by all one vector IQP extended}
\end{figure}

\begin{algorithm}
    \caption{Our hypothesis test protocol}
    \label{alg:Our hypothesis test protocol}  
    \textbf{Input:} $n_{a}$ prime such that $n_{a} + 1$ is a multiple of $8$.
    
    \textbf{Client output:} $o \in \left\{ 0 , 1 \right\}$ 
  
    \textbf{Protocol:}
  
    \begin{algorithmic}[1]
        \STATE Set $n_{p} = \frac{n_{a} + 1}{2}$
        \STATE \label{alg line: Our hypothesis test protocol - start with QR} Take the quadratic residue code generator matrix $\mathbf{Q_{r}} \in \left\{ 0 , 1 \right\}^{n_{a} \times \left( n_{p} - 1 \right)}$ 
        \STATE \label{alg line: Our hypothesis test protocol - add column of ones}Let $\mathbf{Q_{s}} \in \left\{ 0 , 1 \right\}^{n_{a} \times n_{p}}$ be $\mathbf{Q_{r}}$ with a column of ones appended to the last column. 
        \STATE Pick $\widehat{\mathbf{s}} \in \left\{ 0 , 1 \right\}^{n_p - 1}$ chosen uniformly at random. 
        \STATE \label{alg line: Our hypothesis test protocol - randomisation matrix} Define the matrix $\mathbf{A} \in \left\{ 0 , 1 \right\}^{n_p \times n_p}$ according to equation \eqref{equ:transformation matrix}.

        \begin{equation}
            \label{equ:transformation matrix}
            \mathbf{A}_{i,j} =  \begin{cases} 
                                    1 & \text{if } i = j \\ 
                                    0 & \text{if } i \neq j \text{ and } j < n_{p} \\ 
                                    \widehat{\mathbf{s}}_{i} & \text{if } j = n_{p} \text{ and } i < n_p
                                \end{cases}
        \end{equation}
    
        \STATE Set $\mathbf{Q} = \mathbf{Q_{s}} \mathbf{A}$ and $\theta = \frac{\pi}{8}$.
        \STATE Set $\widetilde{\mathbf{Q}}$ to be the matrix $\mathbf{Q_{r}}$ with a column of $-1$ appended.
        \STATE Set $\mathcal{Q}$ to be the uniform distribution over all possible $\mathbf{Q}$ for different $\widehat{\mathbf{s}}$.
        \STATE Perform the IQP computation $\mathbf{Q}$ using Algorithm \ref{alg:real IQP resource honest server} with inputs $\mathbf{Q}$, $\widetilde{\mathbf{Q}}$, $\mathcal{Q}$ and $\theta$ and outputs $\widetilde{x}$ and  $\rho_{B}'$.
        \STATE Let $\mathbf{s} \in \left\{ 0 , 1 \right\}^{n_p}$ be the vector with entries all equal to zero with the exception of the last which is set to one.
        \STATE \label{alg line:Our hypothesis test protocol - test orthogonality} Test the orthogonality of the output $\widetilde{x}$ against $A^{-1} \mathbf{s}^{T}$ setting $o = 0$ if it is not orthogonal and $o = 1$ if it is orthogonal.
    \end{algorithmic}
 
\end{algorithm}

\section{Conclusion and Future Work}

We have presented a protocol that can be used by a limited quantum Client, able to prepare one-qubit Pauli operator eigenstates, to delegate the construction of IQP circuits to a powerful quantum Server. 
By giving the Client of the computation limited quantum abilities (i.e. manipulation of single qubits), we have managed to remove the computational restriction of the Server required in previous work \cite{Temporally_Unstructured_Quantum_Computation}, and therefore have proven information-theoretical security against a malicious Server. The protocol is also proven to be composable and therefore can be used to verify an IQP machine as part of a larger delegated computation.

IQP circuits are also important because they are relatively easy to implement in an experimental setup in comparison to fully fledged quantum computers needed for universal computations. Our protocol requires two layers of measurements, in order to do the appropriate corrections resulting from the blind creation of the state at the Server's side, and for a small number of qubits, it can be implemented even with present technology. A future avenue of research would therefore be the study of this protocol under realistic experimental errors in view of a potential implementation.
	

\section{Acknowledgements}

The authors would like to thank Andru Gheorghiu and Petros Wallden for many enlightening discussions and feedback. This work was supported by grant EP/L01503X/1 for the University of Edinburgh School of Informatics Centre for Doctoral Training in Pervasive Parallelism (http://pervasiveparallelism.inf.ed.ac.uk/) from the UK Engineering and Physical Sciences Research Council (EPSRC) and by grants  EP/K04057X/2, EP/N003829/1 and EP/ M013243/1, as well as by the European Union’s Horizon 2020 Research and Innovation program under Marie Sklodowska-Curie Grant Agreement No. 705194.


\appendix
	
\section{Appendix}

\subsection{Detailed proof of Lemma \ref{lem:state bridge and break correctness}}
\label{apen:Detailed proof of Lemma state bridge and break correctness}
	


\begin{lemma}
	\label{lem:state bridge and break correctness detailed}
	Consider a quantum state $E_{\mathbf{Q}}\ket{\phi}$ where $\ket{\phi}$ is arbitrary. If $\widetilde{\mathbf{Q}}$ is an extended IQP graph built from $\mathbf{Q}$ then there exists a state $E_{\widetilde{\mathbf{Q}}}\ket{\psi}$, which can be transformed into the state $E_{\mathbf{Q}}\ket{\phi}$ through a sequence of Pauli-$Y$ basis measurements on qubits and local rotations around the Z axis on the unmeasured qubits through angles $\left\{ 0 , \frac{\pi}{2} , \pi , \frac{3 \pi}{2} \right\}$.
	
	
	
\end{lemma}

\begin{proof}

	This proof is by construction. We will define a scheme for building the state $\ket{\psi}$ and a corresponding IQP extended graph $\widetilde{\mathbf{Q}}$ meeting the conditions of the lemma.
	
	We begin by considering the case where $\widetilde{\mathbf{Q}}$ was built from the graph $\mathbf{Q}$ by replacing one of the entries, say $\left( i , j \right)$, of $\mathbf{Q}$ with $-1$. $\mathbf{Q}$ can be built from $\widetilde{\mathbf{Q}}$ either by applying a break operation to the vertex $b_{1 = g \left( i , j \right)}$, or by applying a bridge operation to this same vertex. 
	
	We now move to consider these two separate cases.
	
	\begin{itemize}
		\item \textbf{Break:} $\mathbf{Q}_{i j} = 0$. Define the state $E_{\widetilde{\mathbf{Q}}} \ket{\psi}$ as below.
		\begin{equation}
			\label{equ:initial state definition break}
			E_{\widetilde{\mathbf{Q}}} \ket{\psi} = cZ_{a_{i},b_{1}} cZ_{p_{j},b_{1}} E_{\mathbf{Q}} \ket{\phi} \ket{b_1}
		\end{equation}
		Here we set $\ket{b_1} = \ket{r^{b}_{1}}$ with $r^{b}_{1} \in \left\{ 0,1 \right\}$, and $cZ_{p_{j},b_{1}}$ and $cZ_{a_{i},b_{1}}$ are the controlled operators on the respective qubits. Notice then that $cZ_{a_{i},b_{1}} cZ_{p_{j},b_{1}} E_{\mathbf{Q}}$ indeed describes the same entanglement pattern as $E_{\widetilde{\mathbf{Q}}}$. 
		
		Applying the controlled-$Z$ operations is equivalent to applying the operator $Z^{r^{b}_{1}}$ to each of the qubits $a_i$ and $p_j$. We can conclude:
		
		\begin{equation}
			\label{equ:initial state definition break simplfied}
			E_{\widetilde{\mathbf{Q}}} \ket{\psi} = Z^{r^{b}_{1}}_{a_i} Z^{r^{b}_{1}}_{p_j} E_{\mathbf{Q}} \ket{\phi} \ket{b_1}
		\end{equation}
		
		Measuring the qubit $b_1$ in the Pauli-$Y$ basis causes a collapse to either of the Pauli-$Y$ basis states with equal likelihood. It does not, however, have any other effect on the state as the qubit $b_1$ is disentangled. We can therefore discard it and be left with the state $Z^{r^{b}_{1}}_{a_i} Z^{r^{b}_{1}}_{p_j} E_{\mathbf{Q}} \ket{\phi}$ which differs from $E_{\mathbf{Q}} \ket{\phi}$ only by local rotations about the $Z$ axis.
		
		\item \textbf{Bridge:} $\mathbf{Q}_{i j} = 1$. Define the state $E_{\widetilde{\mathbf{Q}}} \ket{\psi}$ as below.
		\begin{equation}
			\label{equ:initial state definition bridge}
			E_{\widetilde{\mathbf{Q}}} \ket{\psi} = cZ_{a_i,b_1} cZ_{p_j,b_1} cZ_{a_i,p_j} E_{\mathbf{Q}} \ket{\phi} \ket{b_1}
		\end{equation}
		Here $\ket{b_1} = \frac{1}{\sqrt{2}} \left( \ket{0} + \left( -1 \right)^{r^{b}_{1}}\ket{1} \right)$ with $r^{b}_{1} \in \left\{ 0 , 1 \right\}$ (i.e a Hadamard basis state). Notice that $cZ_{a_i,b_1} cZ_{p_j,b_1} cZ_{a_i,p_j} E_{\mathbf{Q}}$ describes the same operation as $E_{\widetilde{\mathbf{Q}}}$. 
		
		Applying the operations, $cZ_{a_i,b_1}$ and $cZ_{p_j,b_1}$ to the state $cZ_{a_i,p_j} E_{\mathbf{Q}} \ket{\phi} \ket{ b_1 }$ is equivalent to applying the following operator to the state $cZ_{a_i,p_j} E_{\mathbf{Q}} \ket{\phi}$.
		
		\begin{equation}
			\label{equ:break initial operator}
			\frac{1}{\sqrt{2}} \ket{0} \otimes \mathbb{I}_{a_i} \otimes \mathbb{I}_{p_j} + \left( -1 \right)^{r^{b}_{1}} \frac{1}{\sqrt{2}} \ket{1} \otimes Z_{a_{i}} \otimes Z_{p_{j}}
		\end{equation} 
		
		The above process followed by a measurement of qubit $b_1$ in the Pauli-$Y$ basis is equivalent to applying the following operator to $cZ_{a_i,p_j} E_{\mathbf{Q}} \ket{\phi}$. 
		\begin{equation}
			\label{equ:final bridge operator}
			\frac{1}{\sqrt{2}} \mathbb{I}_{a_{i}} \otimes \mathbb{I}_{p_{j}} + (-1)^{1-s^{b}_{1}} \left( -1 \right)^{r^{b}_{1}} i \frac{1}{\sqrt{2}} Z_{a_{i}} \otimes Z_{p_{j}}
		\end{equation}
		Here we have used the notation that $s^{b}_{1}=0$ when $\ket{+^Y}$ is measured and $s^b_{1}=1$ when $\ket{-^Y}$ is measured. The expression results from post multiplication of expression \eqref{equ:break initial operator} by $\frac{1}{\sqrt{2}} \bra{0} + \left( -1 \right)^{1-s^{b}_{1}} i \frac{1}{\sqrt{2}} \bra{1}$, the conjugate of the Pauli-$Y$ basis states, followed by the appropriate normalisation. The original state of equation \eqref{equ:initial state definition bridge} is transformed, by this measurement, to the state:

		
		\begin{equation}
			\label{equ:post measuremnt state bridge}
			\left( \frac{1}{\sqrt{2}} \mathbb{I}_{a_{i}} \otimes \mathbb{I}_{p_{j}} +  (-1)^{1-s^{b}_{1}} \left( -1 \right)^{r^{b}_{1}} i \frac{1}{\sqrt{2}} Z_{a_{i}} \otimes Z_{p_{j}} \right) cZ_{a_i,p_j} E_{\mathbf{Q}} \ket{\phi}
		\end{equation}
		Notice that the controlled-$Z$ operator can be written as:
		\begin{equation}
			\label{equ:controlled Z}
			cZ_{1,2} = \frac{1}{2} \left( \mathbb{I}_{1} \otimes \mathbb{I}_{2} + Z_{1} \otimes \mathbb{I}_{2} + \mathbb{I}_{1} \otimes Z_{2} + Z_{1} \otimes Z_{2} \right).
		\end{equation}
		Using this fact allows us to see:
		\begin{equation}
			\label{equ:controlled Z S}
			cZ_{1,2} = \left( S_{1} \otimes S_{2} \right) \left( \frac{1}{\sqrt{2}} \mathbb{I}_1 \otimes \mathbb{I}_2 + i \frac{1}{\sqrt{2}} Z_1 \otimes Z_2 \right)
		\end{equation}
		\begin{equation}
			\label{equ:controlled Z S minus}
			cZ_{1,2} = \left( S^{-1}_{1} \otimes S^{-1}_{2} \right) \left( \frac{1}{\sqrt{2}} \mathbb{I}_1 \otimes \mathbb{I}_2 - i \frac{1}{\sqrt{2}} Z_1 \otimes Z_2 \right)
		\end{equation}
		
		In particular:
		\begin{multline}
			\label{equ:conrolled Z general}
			cZ_{a_i,p_j} = \left( S_{a_i}^{(-1)^{1-s^{b}_{1}} \left( - \left( -1 \right)^{r^{b}_{1}} \right)} \otimes S_{p_j}^{(-1)^{1-s^{b}_{1}} \left( - \left( -1 \right)^{r^{b}_{1}} \right)} \right) \\ \left( \frac{1}{\sqrt{2}} \mathbb{I}_{a_i} \otimes \mathbb{I}_{p_j} - (-1)^{1-s^{b}_{1}}\left( -1 \right)^{r^{b}_{1}} i \frac{1}{\sqrt{2}} Z_{a_i} \otimes Z_{p_j} \right)
		\end{multline}
		
	Substituting this into \eqref{equ:post measuremnt state bridge}, and with some rearranging, we realise the resulting state is actually that of equation \eqref{equ:bridge final state}.
		
		\begin{equation}
			\label{equ:bridge final state}
			\left( S_{a_{1}}^{(-1)^{s^{b}_{1}+r^{b}_{1}}} \otimes S_{p_{j}}^{(-1)^{s^{b}_{1}+r^{b}_{1}}} \right) E_{\mathbf{Q}} \ket{\phi}
		\end{equation}
		
	Once again this differs from the state $E_{\mathbf{Q}} \ket{\phi}$ only by local rotations around the $Z$ axis.

	\end{itemize}
	
	We now turn to the case where the number of break and bridge operations needed to move from $\widetilde{\mathbf{Q}}$ to $\mathbf{Q}$ is more than one. The state $E_{\widetilde{\mathbf{Q}}} \ket{\psi}$ can be built one step at a time by repeating the steps above (i.e. entangling the appropriate bridge and break qubits one at a time). The proof that the state resulting from measurements of the qubits $b_{1} , ... , b_{m}$ would  result in the graph $E_{\mathbf{Q}} \ket{\phi}$ follows for the following reasoning. Since the qubits that might require corrections are never measured, all measurements and corrections commute. The entanglement operators too commute with the corrections and the measurement operations when they do not act upon the same qubits. As such all operations commute, therefore all the necessary entanglement operations, measurement operations and all the necessary corrections can be done in this order, all at once.
\end{proof}
	
\begin{algorithm}
  \caption{IQP computation}
  \label{alg:x-prog in MBQC}  
  \textbf{Input:} $\mathbf{Q} \in \left\{ 0 , 1 \right\}^{n_{a} \times n_{p}}$, $\theta$
    
  \textbf{Output:} $x \in \left\{ 0 , 1 \right\}^{n_{p}}$
  
  \textbf{Protocol:}
  
  \begin{algorithmic}[1]
    \STATE Generate states $\ket{+} = \ket{p_{j}}$ and $\ket{+} = \ket{a_{i}}$ for $j \in \left\{ 0 , ... n_{p} \right\}$ and $i \in \left\{ 0 , ... n_{a} \right\}$.
    \STATE Implement the operations $E_{\mathbf{Q}}$ on the generated qubits.
    \STATE Measure primary qubits in the Hadamard basis and ancillary qubits in the basis of equation \eqref{equ:IQP measuremnt basis} to obtain measurement outcomes $s^{p} \in \left\{ 0 , 1 \right\}^{n_{p}}$ and $s^{a} \in \left\{ 0 , 1 \right\}^{n_{a}}$.
    \STATE Perform corrections according to equation \eqref{equ:IQP final outcome calculation original} to generate output $\widetilde{x}$.
    \begin{equation}
      \label{equ:IQP final outcome calculation original}
      \widetilde{x}_{j} = s^{p}_{j} + \sum_{i:\mathbf{Q}_{ i j } = 1}  s^{a}_{i} \pmod 2
    \end{equation}
 \end{algorithmic}
 
\end{algorithm}

\begin{algorithm}
  \caption{Distributed IQP computation}
  \label{alg:real IQP resource honest server distributed}
  \textbf{Public:} $\widetilde{\mathbf{Q}} , \mathcal{Q}$, $\theta$
  
  \textbf{Client input:} $\mathbf{Q} \in \left\{ 0 , 1 \right\}^{n_{a} \times n_{p}}$ 
  
  \textbf{Client output:} $\widetilde{x}$
  
  \textbf{Protocol:}
  
  \begin{algorithmic}[1]
    \STATE The Client generates the states $\ket{+} = \ket{p_{j}}$ and $\ket{+} = \ket{a_{i}}$ for $j \in \left\{ 0 , ... n_{p} \right\}$ and $i \in \left\{ 0 , ... n_{a} \right\}$
    \label{alg line:real IQP resource honest server distributed - primary and ancillary state generation}
    \STATE Client creates $d^b \in \left\{0,1\right\}^{n_b}$ in the following way: For $i=1,\dots,n_a$ and $j=1,\dots,n_p$, if $\widetilde{\mathbf{Q}}_{ij}=-1$ and $\mathbf{Q}_{ij}=0$, then $d^b_k=0$ else if $\widetilde{\mathbf{Q}}_{ij}=-1$ and $\mathbf{Q}_{ij}=1$ then $d^b_k=1$. He keeps track of the relation between $k$ and $(i,j)$ via the surjective function $g: \mathbb{Z}_{n_a \times n_p} \rightarrow \mathbb{Z}_{n_b}$.  
    \STATE The Client generates $r^{b} \in \left\{ 0 , 1 \right\}^{n_{b}}$ at random and produces the states $\ket{b_{k}} = Y^{r^{b}_{k}} \left( \sqrt{Y} \right)^{d^{b}_{k}} \ket{0} $ for $k \in \left\{ 1 , ... , n_{b} \right\}$
    \label{alg line:real IQP resource honest server distributed - bridge and break state generation}
    \STATE State $\rho$ comprising of all of the Client's produced states is sent to the Server.
    \STATE The Server implements $E_{\widetilde{\mathbf{Q}}}$.
    \STATE The Server measures qubits $b_{1} , ... , b_{n_{b}}$ in the $Y$-basis $\left\{ \ket{+^{Y}} , \ket{-^{Y}} \right\}$ and sends the outcome $s^{b} \in \left\{ 0 , 1 \right\}^{n_{b}}$ to the Client.
    \STATE The Client calculates $\Pi^{z} , \Pi^{s} \in \left\{ 0 , 1 \right\}^{n_{p}}$ and $A^{z} , A^{s} \in \left\{ 0 , 1 \right\}^{n_{a}}$ using equations \eqref{equ:primary Z correction term not blind} - \eqref{equ:ancila S correction term not blind}. 
    \begin{align}
      \label{equ:primary Z correction term not blind}
      \Pi^{z}_{j} &= \sum_{i,k:g(i,j)=k} r_k^b \left( 1 - d^{b}_k \right) \\
      \label{equ:primary S correction term not blind}
      \Pi^{s}_{j} &= \sum_{i,k:g(i,j)=k} (-1)^{s^{b}_k+r^{b}_k} d^{b}_k  \\
      \label{equ:ancila Z correction term not blind}
      A^{z}_{i} &= \sum_{j,k:g(i,j)=k} r_k^{b} \left( 1 - d^{b}_k \right) \\
      \label{equ:ancila S correction term not blind}
      A^{s}_{i} &= \sum_{j,k:g(i,j)=k}(-1)^{s^{b}_k+r^{b}_k}  d^{b}_k 
    \end{align}
    \STATE The Client sends $A\in\{0,1,2,3\}^{n_a}$ and $\Pi\in\{0,1,2,3\}^{n_p}$ for the ancillary and primary qubits respectively, where $A_{i} = A^{s}_{i} + 2 A^{z}_{i} \pmod 4$ and $\Pi_{j} = \Pi^{s}_{j} + 2 \Pi^{z}_{j} \pmod 4$.
    \STATE The Server measures their qubits in the basis below, for the ancillary and primary qubits respectively. 
    \begin{equation}
        \label{equ:primary and ancillary measurement basis original}
        S^{- A_{i} } \left\{ \ket{0_\theta} , \ket{1_\theta} \right\}  \text{ and } S^{- \Pi_{j} } \left\{ \ket{+} , \ket{-} \right\}
    \end{equation}
    The measurement outcomes $s^{p} \in \left\{ 0 , 1 \right\}^{n_{p}}$ and $s^{a} \in \left\{ 0 , 1 \right\}^{n_{a}}$ are sent to the Client.
    \STATE The Client generates and outputs $\widetilde{x} \in \left\{ 0 , 1 \right\}^{n_{p}}$ using equation \eqref{equ:IQP final outcome calculation}
 \end{algorithmic}
 
\end{algorithm}

\ANNASCOMMENT{
\begin{algorithm}
  \caption{d calculation}
  \label{alg:d calculation}
  \textbf{Public:} $\widetilde{\mathbf{Q}}$
  
  \textbf{Client input:} $\mathbf{Q} \in \left\{ 0 , 1 \right\}^{n_{a} \times n_{p}}$
  
  \textbf{Output:} $d \in \left\{ 0 , 1 \right\}^{n_{b}}$
  
  \textbf{Protocol:}
  
  \begin{algorithmic}[1]
    \FOR{$j \in \left\{1 , \dots , n_{p} \right\}$} 
      \FOR{$i \in \left\{1 , \dots , n_{a} \right\}$} 
	\IF{$\widetilde{\mathbf{Q}}_{ij} = -1$}
	  \STATE $d_{ij}  =  \mathbf{Q}_{ij} $
	\ENDIF
      \ENDFOR 
    \ENDFOR
 \end{algorithmic}
 
\end{algorithm}
}

\subsection{Extended Proof of Theorem \ref{thm:security proof}}

\begin{theorem}
    \label{thm:security of real resource appendix}
    The protocol described by Algorithm \ref{alg:real IQP resource honest server} is secure against a dishonest Server. 
\end{theorem}

\begin{proof}
    The proof consists of a pattern of transformations of the real protocol of Algorithm \ref{alg:real IQP resource honest server}, into the ideal resource of Algorithm \ref{alg:real IQP resource honest server with simulator}, which leaves the computation unchanged, therefore ensuring the indistinguishability of the two settings. 
  
    The first transformation we perform is of the state generation phase of the Algorithm \ref{alg:real IQP resource honest server}. The new method we use for this phase is described in Algorithm \ref{alg:real IQP resource honest server with added teleportation} and relies on the measurement of EPR pairs to produce qubits in the correct basis, with some randomness resulting from the measurement. This may be visualised by the expansion of $\pi_{A}^{1}$ seen in Figure \ref{fig:pre generated randomness resource} in the main text.
  
    While lines \ref{alg line:real IQP resource honest server - primary and ancila random key generation}, \ref{alg line:real IQP resource honest server - primary and ancillary state generation} and \ref{alg line:real IQP resource honest server - bridge and break state generation} of Algorithm \ref{alg:real IQP resource honest server} and the lines \ref{alg line:real IQP resource honest server with added teleportation - primary and ancila random key generation}, \ref{alg line:real IQP resource honest server with added teleportation - EPR pair generation}, \ref{alg line:real IQP resource honest server with added teleportation - primary and ancila EPR measurment} and \ref{alg line: real IQP resource honest server with added teleportation - bridge and break EPR measuremnt} of Algorithm \ref{alg:real IQP resource honest server with added teleportation} differ the remainder of both algorithms is identical. We show now that the algorithms are indistinguishable. 
  
    Firstly consider the generation of $r^{p}$ and $r^{a}$. In Algorithm \ref{alg:real IQP resource honest server} these terms are picked uniformly at random from the set of all binary stings of the appropriate length. In the case of Algorithm \ref{alg:real IQP resource honest server with added teleportation} they are generated by measurements on EPR pairs, the result of which is entirely random. Similarly, in both cases, $r^{b}$ is picked uniformly at random from the set of all binary strings of the appropriate length.
  
    Line \ref{alg line:real IQP resource honest server - primary and ancillary state generation} of Algorithm \ref{alg:real IQP resource honest server} generates at random one of the four states $\ket{+}$, $\ket{+^{Y}}$, $\ket{-}$ and $\ket{-^{Y}}$. Line \ref{alg line:real IQP resource honest server with added teleportation - primary and ancila EPR measurment} of Algorithm \ref{alg:real IQP resource honest server with added teleportation} achieves the same effect by measuring an EPR pair with equal probability in one of the basis $\left\{ \ket{+} , \ket{-} \right\}$ and $\left\{ \ket{+^{Y}} , \ket{-^{Y}} \right\}$.
  
    Finally, the application of the $\left( \sqrt{Y} \right)^{d^{b}_{j}}$ operation in line \ref{alg line:real IQP resource honest server - bridge and break state generation} of Algorithm \ref{alg:real IQP resource honest server} decides, according the graph to be created, if the bridge and break qubit will be drawn from the set $\left\{ \ket{+} , \ket{-} \right\}$ or $\left\{ \ket{0} , \ket{1} \right\}$. Choosing between using the measurement basis $\left\{ \ket{+} , \ket{-} \right\}$ or $\left\{ \ket{0} , \ket{1} \right\}$ on one half of an EPR pair of course has the same effect. The random rotation $Y^{r^{b}}_{j}$ then has the same effect of the randomness that is intrinsic to the measurement performed in Algorithm \ref{alg:real IQP resource honest server with added teleportation}.
  
  
    Consider now the transformation from Algorithm \ref{alg:real IQP resource honest server with added teleportation} to Algorithm \ref{alg:real IQP resource honest server with added teleportation, rearrangement and pre-made randomness}. Notice that line \ref{alg line:real IQP resource honest server with added teleportation - primary and ancila EPR measurment} of Algorithm \ref{alg:real IQP resource honest server with added teleportation} is identical to that of line \ref{alg line:real IQP resource honest server with added teleportation, rearrangement and pre-made randomness - primary and ancila EPR measument} of Algorithm \ref{alg:real IQP resource honest server with added teleportation, rearrangement and pre-made randomness}. This operation can be delayed without affecting the computation as the qubit being measured is not acted upon in any other way during the protocol.
  
    Consider $\Pi$ and $A$. In Algorithm \ref{alg:real IQP resource honest server with added teleportation, rearrangement and pre-made randomness} they are generated at random from the set of all $\Pi \in \left[ 0 , 1 , 2 , 3 \right]^{n_{p}}$ and $A \in \left[ 0 , 1 , 2 , 3 \right]^{n_{a}}$ as stated in line \ref{alg line:real IQP resource honest server with added teleportation, rearrangement and pre-made randomness - primary and ancila random key generation}. This is the case too for Algorithm \ref{alg:real IQP resource honest server with added teleportation} because $\Pi^{z}_{i}$, $\Pi^{s}_{i}$, $A^{z}_{k}$ and $A^{s}_{k}$ are one time padded by $r^{p}_{i}$, $d^{p}_{i}$, $r^{a}_{k}$ and $d^{a}_{k}$ respectively as seen in equations \eqref{equ:primary Z correction term}, \eqref{equ:primary S correction term}, \eqref{equ:ancila Z correction term} and \eqref{equ:ancila S correction term}. 
  
    It remains to show that Algorithm \ref{alg:real IQP resource honest server with added teleportation, rearrangement and pre-made randomness} results in the same computation as Algorithm \ref{alg:real IQP resource honest server with added teleportation}. This can be achieved by noting a simple rearrangement of equations \eqref{equ:primary Z correction term}, \eqref{equ:primary S correction term}, \eqref{equ:ancila Z correction term} and \eqref{equ:ancila S correction term} to make $d^{p}_{i}$ and $d^{a}_{k}$ the subject. In doing so we assume the $r_{p}^{i} , r_{a}^{k} = 0$ which is corrected for, if this is not the case, in equation \eqref{equ:IQP final outcome calculation with corrections}. The reader may wish to refer to Figure \ref{fig:pre generated randomness resource} in the main text for a visualisation of this new resource.
  
    Finally, Algorithm \ref{alg:real IQP resource honest server with simulator} simply involves a relabeling of the players in the protocol of Algorithm \ref{alg:real IQP resource honest server with added teleportation, rearrangement and pre-made randomness} to match those in the simulator distinguisher setting. This amounts to the transformation from Figure \ref{fig:pre generated randomness resource} to Figure \ref{fig:rearranged into simulator setting} in the main text.
    
    This series of transformations convinces us that the following relationship is true and that the resource of Algorithm \ref{alg:real IQP resource honest server} is composably secure against a dishonest Server.
    
    \begin{equation}
        \pi_A\mathcal{R}\equiv \mathcal{S}\sigma
    \end{equation}
  
\end{proof}

\begin{algorithm}
  \caption{Blind distributed IQP computation with teleportation technique}
  \label{alg:real IQP resource honest server with added teleportation}
  \textbf{Public:} $\widetilde{\mathbf{Q}} , \mathcal{Q}$, $\theta$
  
  \textbf{Client input:} $\mathbf{Q}$ 
  
  \textbf{Client output:} $\widetilde{x}$ 
  
  \textbf{Protocol:}
  
  \begin{algorithmic}[1]
    \STATE The Client randomly generates $d^{p} \in \left\{ 0 , 1 \right\}^{n_{p}}$ and $d^{a} \in \left[ 0 , 1 \right]^{n_{a}}$ where $n_{p}$ and $n_{a}$ are the numbers of primary and ancillary qubits respectively.
    \label{alg line:real IQP resource honest server with added teleportation - primary and ancila random key generation}
    \STATE The Client generates $n_{p}$ EPR pairs $\ket{EPR^{p}_{j}}$, $n_{a}$ EPR pairs $\ket{EPR^{a}_{i}}$ and a further $n_{b}$ EPR pairs $\ket{EPR^{b}_{k}}$.
    \label{alg line:real IQP resource honest server with added teleportation - EPR pair generation}
    \STATE The Client measures one half of each of $\ket{EPR^{p}_{j}}$ in the basis $S^{d^{p}_{j}} \left\{ \ket{+} , \ket{-} \right\}$ to achieve outcome $r^{p}_{j}$ and one half of each of $\ket{EPR^{a}_{i}}$ in the basis $S^{d^{a}_{i}} \left\{ \ket{+} , \ket{-} \right\}$ to achieve outcome $r^{a}_{i}$.
    \label{alg line:real IQP resource honest server with added teleportation - primary and ancila EPR measurment}
    \STATE Client creates $d^b \in \left\{0,1\right\}^{n_b}$ in the following way: For $i=1,\dots,n_a$ and $j=1,\dots,n_p$, if $\widetilde{\mathbf{Q}}_{ij}=-1$ and $\mathbf{Q}_{ij}=0$, then $d^b_k=0$ else if $\widetilde{\mathbf{Q}}_{ij}=-1$ and $\mathbf{Q}_{ij}=1$ then $d^b_k=1$. He keeps track of the relation between $k$ and $(i,j)$ via the surjective function $g: \mathbb{Z}_{n_a \times n_p} \rightarrow \mathbb{Z}_{n_b}$.  
    \STATE The Client measures one half of each of $\ket{EPR^{b}_{k}}$ in the basis $\sqrt{Y}^{d^{b}_{k}} \left\{ \ket{0} , \ket{1} \right\}$ to achieve outcome $r^{b}_{k}$.
    \label{alg line: real IQP resource honest server with added teleportation - bridge and break EPR measuremnt}
    \STATE State $\rho$ comprising of all unmeasured states in the Client's position is sent to the server.
    \label{alg line:real IQP resource honest server with added teleportation - sending measured states}
    \STATE The Server implements $E_{\widetilde{\mathbf{Q}}}$.
    \STATE The Server measures qubits $b_{1} , ... , b_{n_{b}}$ in the $Y$-basis $\left\{ \ket{+^{Y}} , \ket{-^{Y}} \right\}$ and sends the outcome $s^{b} \in \left\{ 0 , 1 \right\}^{n_{b}}$ to the Client.
    \STATE The Client calculates $\Pi^{z} , \Pi^{s} \in \left\{ 0 , 1 \right\}^{n_{p}}$ and $A^{z} , A^{s} \in \left\{ 0 , 1 \right\}^{n_{a}}$ using equations \eqref{equ:primary Z correction term}, \eqref{equ:primary S correction term}, \eqref{equ:ancila Z correction term} and \eqref{equ:ancila S correction term}.
    \STATE The Client sends $A\in\{0,1,2,3\}^{n_a}$ and $\Pi\in\{0,1,2,3\}^{n_p}$ for the ancillary and primary qubits respectively, where $A_{i} = A^{s}_{i} + 2 A^{z}_{i} \pmod 4$ and $\Pi_{j} = \Pi^{s}_{j} + 2 \Pi^{z}_{j} \pmod 4$.
    \STATE The Server measures their qubits in the two basis of equation \eqref{equ:primary and ancillary measurement basis} for the ancillary and primary qubits respectively. The measurement outcomes $s^{p} \in \left\{ 0 , 1 \right\}^{n_{p}}$ and $s^{a} \in \left\{ 0 , 1 \right\}^{n_{a}}$ are sent to the Client.
    \STATE The Client generates and outputs $\widetilde{x} \in \left\{ 0 , 1 \right\}^{n_{p}}$ using equation \eqref{equ:IQP final outcome calculation}
  \end{algorithmic}
 
\end{algorithm}

\begin{algorithm}
  \caption{Blind distributed IQP computation with teleportation technique, rearrangement and pre-made randomness}
  \label{alg:real IQP resource honest server with added teleportation, rearrangement and pre-made randomness}
  \textbf{Public:} $\widetilde{\mathbf{Q}} , \mathcal{Q}$ , $\theta$
  
  \textbf{Client input:} $\mathbf{Q}$ 
  
  \textbf{Client output:} $\widetilde{x}$ 
  
  \textbf{Protocol:}
  
  \begin{algorithmic}[1]
    \STATE The Client generates $n_{p}$ EPR pairs $\ket{EPR^{p}_{j}}$, $n_{a}$ EPR pairs $\ket{EPR^{a}_{i}}$ and a further $n_{b}$ EPR pairs $\ket{EPR^{b}_{k}}$.
    \STATE Half of each EPR pair is sent, by the Client, to the Server.
    \STATE Client creates $d^b \in \left\{0,1\right\}^{n_b}$ in the following way: For $i=1,\dots,n_a$ and $j=1,\dots,n_p$, if $\widetilde{\mathbf{Q}}_{ij}=-1$ and $\mathbf{Q}_{ij}=0$, then $d^b_k=0$ else if $\widetilde{\mathbf{Q}}_{ij}=-1$ and $\mathbf{Q}_{ij}=1$ then $d^b_k=1$. He keeps track of the relation between $k$ and $(i,j)$ via the surjective function $g: \mathbb{Z}_{n_a \times n_p} \rightarrow \mathbb{Z}_{n_b}$.  
    \STATE The Client measures one half of each of $\ket{EPR^{b}_{k}}$ in the basis $\sqrt{Y}^{d^{b}_{k}} \left\{ \ket{0} , \ket{1} \right\}$ to achieve outcome $r^{b}_{k}$.
    \STATE The Server implements $E_{\widetilde{\mathbf{Q}}}$.
    \STATE Qubits $b_{1} , ... , b_{n_{b}}$ are measured by the Server in the $y$-basis $\left\{ \ket{+^{Y}} , \ket{-^{Y}} \right\}$, producing the outcome $s_{b} \in \left\{ 0 , 1 \right\}^{n_{b}}$ which are returned to the Client.
    \STATE The Client randomly generated $\Pi \in \left\{ 0 , 1 , 2 , 3 \right\}^{n_{p}}$ and $A \in \left[ 0 , 1 , 2 , 3 \right]^{n_{a}}$ where $n_{p}$ and $n_{a}$ are the numbers of primary and ancillary qubits respectively.
    \label{alg line:real IQP resource honest server with added teleportation, rearrangement and pre-made randomness - primary and ancila random key generation}
    \STATE The Client calculates $d^{p}_{j} \in \left\{ 0 , 1 , 2 , 3 \right\}^{n_{p}}$ and $d^{a}_{i} \in \left\{ 0 , 1 , 2 , 3 \right\}^{n_{a}}$ using equations \eqref{equ:primary measurement term} and \eqref{equ:ancila measurement term} respectively.
    \STATE The Client measures one half of each of $\ket{EPR^{p}_{j}}$ in the basis $S^{d^{p}_{j}} \left\{ \ket{+} , \ket{-} \right\}$ to achieve outcome $r^{p}_{j}$ and one half of each of $\ket{EPR^{a}_{i}}$ in the basis $S^{d^{a}_{i}} \left\{ \ket{+} , \ket{-} \right\}$ to achieve outcome $r^{a}_{i}$.
    \label{alg line:real IQP resource honest server with added teleportation, rearrangement and pre-made randomness - primary and ancila EPR measument}
    \STATE The Client sends $A\in\{0,1,2,3\}^{n_a}$ and $\Pi\in\{0,1,2,3\}^{n_p}$ for the ancillary and primary qubits respectively.
    \STATE The Server measures their qubits in the two basis of equation \eqref{equ:primary and ancillary measurement basis} for the ancillary and primary qubits respectively. The measurement outcomes $s^{p} \in \left\{ 0 , 1 \right\}^{n_{p}}$ and $s^{a} \in \left\{ 0 , 1 \right\}^{n_{a}}$ are sent to the Client.
    \STATE The Client generates and outputs $x \in \left\{ 0 , 1 \right\}^{n_{p}}$ using equation \eqref{equ:IQP final outcome calculation with corrections}.
  \end{algorithmic}
 
\end{algorithm}

\begin{landscape}

\subsection{Pictorial Evolution of Algorithms in This Paper}
\label{sec:Pictorial Evolution of Algorithms in This Paper}

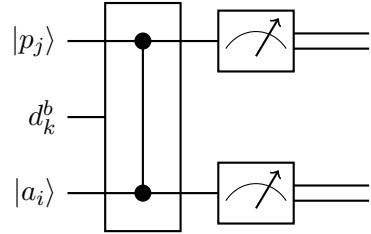
\begin{figure}[H]
  \centering
  \begin{tikzpicture}
    \draw[thick] (1,1) node[anchor=east] {$d^{b}_{k}$} -- (1.5,1);
    \draw[thick] (1,2) node[anchor=east] {$\ket{p_{j}}$} -- (3,2);
    \draw[thick] (1,0) node[anchor=east] {$\ket{a_{i}}$} -- (3,0);
    
    \draw[thick] (1.5,-0.5) rectangle (2.5,2.5);
		
    \filldraw (2,0) circle (3pt);
    \filldraw (2,2) circle (3pt);
    \draw[thick] (2,0) -- (2,2);
    
    \draw[thick] (3,1.6) rectangle (4,2.4);
    \draw (3.1,1.9) .. controls (3.3,2.2) and (3.7,2.2) .. (3.9,1.9);
    \draw[thick, ->] (3.5, 1.8) -- (3.8, 2.3);
    
    \draw[thick] (3,-0.4) rectangle (4,0.4);
    \draw (3.1,-0.1) .. controls (3.3,0.2) and (3.7,0.2) .. (3.9,-0.1);
    \draw[thick, ->] (3.5, -0.2) -- (3.8, 0.3);
    
    \draw[thick] (4,0.1) -- (5,0.1);
    \draw[thick] (4,-0.1) -- (5,-0.1);
    
    \draw[thick] (4,2.1) -- (5,2.1);
    \draw[thick] (4,1.9) -- (5,1.9);
  \end{tikzpicture}
  \caption{Circuit to implement IQP. The controlled-$Z$ is controlled by $d^{b}_{k} = \mathbf{Q}_{i j}$ where $j$ and $i$ are the indices of the primary and ancillary qubits. In other words  $d^{b}_{k} = 1$ means the primary and ancillary qubits are to be entangled. This is the method described in Section \ref{sec:Preliminaries}.}
  \label{fig:original graph generation circuit}
\end{figure}

\begin{figure}[H]
  \centering
  \begin{tikzpicture}
    \draw[thick] (-3,1) node[anchor=east] {$\ket{0}$} -- (-2,1);
    
    \draw[thick] (-2,0.5) rectangle (1,1.5) node[align = center, pos = 0.5] {$\left( \sqrt{Y} \right)^{d^{b}_{k}} Y^{r^{b}_{k}}$} node[anchor = west] {(a)}; 
  
    \draw[thick] (-3,2) node[anchor=east] {$\ket{p_{j}}$} -- (4,2);
    \draw[thick] (1,1) -- (3,1);
    \draw[thick] (-3,0) node[anchor=east] {$\ket{a_{i}}$} -- (4,0);
		
    \filldraw (2,1) circle (3pt);
    \filldraw (2,2) circle (3pt);
    \draw[thick] (2,1) -- (2,2);
    
    \filldraw (2,0) circle (3pt);
    \filldraw (2,1) circle (3pt);
    \draw[thick] (2,0) -- (2,1);

    \draw[thick] (3,0.6) rectangle (4,1.4);
    \draw (3.1,0.9) .. controls (3.3,1.2) and (3.7,1.2) .. (3.9,0.9);
    \draw[thick, ->] (3.5, 0.8) -- (3.8, 1.3);

    \draw[thick] (4,1.6) rectangle (9,2.8) node[align = center, pos = 0.5] {$\left( S^{-(-1)^{s^{b}_{k}+r^{b}_{k}}} \right)^{d^{b}_{k}} \left( Z^{r^{b}_{k}} \right)^{1-d^{b}_{k}}$}node[anchor = west] {(b)};
    \draw[thick] (4,-0.8) rectangle (9,0.4) node[align = center, pos = 0.5] {$\left( S^{-(-1)^{s^{b}_{k}+r^{b}_{k}}} \right)^{d^{b}_{k}} \left( Z^{r^{b}_{k}} \right)^{1-d^{b}_{k}}$}node[anchor = west] {(b)};
    
    \draw[thick] (4,1.1) -- (5.4,1.1) -- (5.4,1.6);
    \draw[thick] (4,0.9) -- (5.4,0.9) -- (5.4,0.4);
    \draw[thick] (5.6,1.6) -- (5.6,0.4);
    
    \draw[thick] (9,2) -- (10,2);
    \draw[thick] (9,0) -- (10,0);
    
    \draw[thick] (10,1.6) rectangle (11,2.4);
    \draw (10.1,1.9) .. controls (10.3,2.2) and (10.7,2.2) .. (10.9,1.9);
    \draw[thick, ->] (10.5, 1.8) -- (10.8, 2.3);
    
    \draw[thick] (10,-0.4) rectangle (11,0.4);
    \draw (10.1,-0.1) .. controls (10.3,0.2) and (10.7,0.2) .. (10.9,-0.1);
    \draw[thick, ->] (10.5, -0.2) -- (10.8, 0.3);
    
    \draw[thick] (11,0.1) -- (12,0.1);
    \draw[thick] (11,-0.1) -- (12,-0.1);
    
    \draw[thick] (11,2.1) -- (12,2.1);
    \draw[thick] (11,1.9) -- (12,1.9);
  \end{tikzpicture}
  \caption{Circuit to implement IQP with additional intermediate qubit. This is the method described in Section \ref{sec:Break, Bridge Operators}. The gate at (a) describes the process of generating the break and bridge qubit while those at (b) display the corrections necessary as a result of the bridge and break process}
  \label{fig:bridge and break graph generation circuit}
\end{figure}
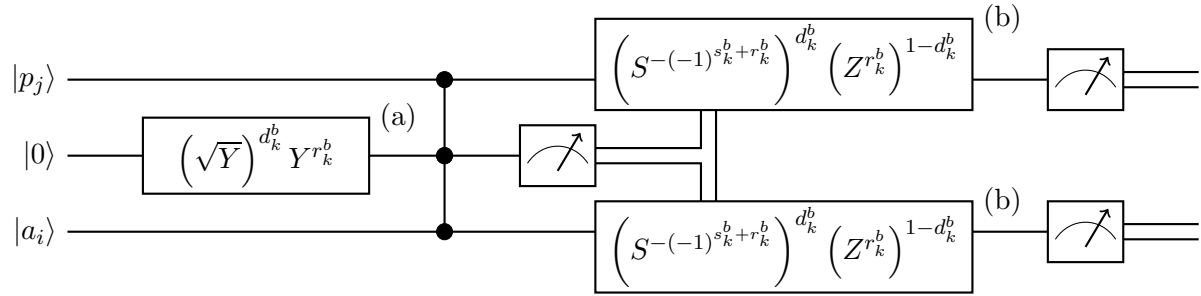

\begin{figure}[H]
  \centering
  \begin{tikzpicture}
    \draw[thick] (-5,2) node[anchor=east] {$\ket{p_{j}}$} -- (-4,2);
    \draw[thick] (-5,1) node[anchor=east] {$\ket{0}$} -- (-4,1);
    \draw[thick] (-5,0) node[anchor=east] {$\ket{a_{i}}$} -- (-4,0);
    
    \draw[thick] (-4,0.5) rectangle (-1,1.5) node[align = center, pos = 0.5] {$\left( \sqrt{Y} \right)^{d^{b}_{k}} Y^{r^{b}_{k}}$}; 
    \draw[thick] (-4,1.6) rectangle (-1,2.4) node[align = center, pos = 0.5] {$S^{d^{p}_{j}} Z^{r^{p}_{j}}$};
    \draw[thick] (-4,-0.4) rectangle (-1,0.4) node[align = center, pos = 0.5] {$S^{d^{a}_{i}} Z^{r^{a}_{i}}$};
    
    \draw[thick] (-1,2) -- (0.5,2);
    \draw[thick] (-1,0) -- (0.5,0);
  
    \draw[thick] (3,2) -- (4,2);
    \draw[thick] (-1,1) -- (3,1);
    \draw[thick] (3,0) -- (4,0);
		
    \filldraw (0,1) circle (3pt);
    \filldraw (0,2) circle (3pt);
    \draw[thick] (0,1) -- (0,2);
     
    \filldraw (0,0) circle (3pt);
    \filldraw (0,1) circle (3pt);
    \draw[thick] (0,0) -- (0,1);

    \draw[thick] (3,0.6) rectangle (4,1.4);
    \draw (3.1,0.9) .. controls (3.3,1.2) and (3.7,1.2) .. (3.9,0.9);
    \draw[thick, ->] (3.5, 0.8) -- (3.8, 1.3);
    
    \draw[thick] (0.5,1.6) rectangle (3,2.6) node[align = center, pos = 0.5] {$\left( S^{d^{p}_{j}} Z^{r^{p}_{j}} \right)^{-1}$};
    \draw[thick] (0.5,-0.6) rectangle (3,0.4) node[align = center, pos = 0.5] {$\left( S^{d^{a}_{i}} Z^{r^{a}_{i}} \right)^{-1}$};
    
    \draw[thick] (4,1.6) rectangle (9,2.8) node[align = center, pos = 0.5] {$\left( S^{-(-1)^{s^{b}_{k}+r^{b}_{k}}} \right)^{d^{b}_{k}} \left( Z^{r^{b}_{k}} \right)^{1-d^{b}_{k}}$};
    \draw[thick] (4,-0.8) rectangle (9,0.4) node[align = center, pos = 0.5] {$\left( S^{-(-1)^{s^{b}_{k}+r^{b}_{k}}} \right)^{d^{b}_{k}} \left( Z^{r^{b}_{k}} \right)^{1-d^{b}_{k}}$};
    
    \draw[thick] (4,1.1) -- (5.4,1.1) -- (5.4,1.6);
    \draw[thick] (4,0.9) -- (5.4,0.9) -- (5.4,0.4);
    \draw[thick] (5.6,1.6) -- (5.6,0.4);
    
    \draw[thick] (9,2) -- (10,2);
    \draw[thick] (9,0) -- (10,0);
    
    \draw[thick] (10,1.6) rectangle (11,2.4);
    \draw (10.1,1.9) .. controls (10.3,2.2) and (10.7,2.2) .. (10.9,1.9);
    \draw[thick, ->] (10.5, 1.8) -- (10.8, 2.3);
    
    \draw[thick] (10,-0.4) rectangle (11,0.4);
    \draw (10.1,-0.1) .. controls (10.3,0.2) and (10.7,0.2) .. (10.9,-0.1);
    \draw[thick, ->] (10.5, -0.2) -- (10.8, 0.3);
    
    \draw[thick] (11,0.1) -- (12,0.1);
    \draw[thick] (11,-0.1) -- (12,-0.1);
    
    \draw[thick] (11,2.1) -- (12,2.1);
    \draw[thick] (11,1.9) -- (12,1.9);
    
    \draw[very thick, dotted] (-4.1,-0.5) rectangle (-0.9,2.5) node[anchor = west] {(c)};
    \draw[very thick, dotted] (0.4,-0.9) rectangle (9.1,0.5) node[anchor = west] {(d)};
  \end{tikzpicture}
  \caption{Circuit to implement IQP with additional intermediate qubit and randomness. This is the method used in Section \ref{sec:security proof}. The dotted box (c) indicates the preparation to be done by the Client while (d) indicates the corrections to be done. These corrections are Incorporated into the measurements.}
\end{figure}
\end{landscape}


\begin{thebibliography}{9}

\bibitem{Simulating Physics with Computers}
	Richard P. Feynman, \emph{Simulating Physics with Computers}, Int. J. Theor. Phys. 21, 467–488 (1982).

\bibitem{Quantum Simulation}
	I. M. Georgescu, S. Ashhab and Franco Nori, \emph{Quantum Simulation}, Reviews of Modern Physics (2014).
			
\bibitem{Quantum Computing and Quantum Information}
	Michael A. Nielsen, Isaac L. Chuang, \emph{Quantum Computing and Quantum Information}, Cambridge University Press New York, NY, USA (2011).
		
\bibitem{Poly-Time Algorithms for Prime Factorisation and Discrete Logarithms on a Quantum Computer}
	Peter W. Shor, \emph{Poly-Time Algorithms for Prime Factorisation and Discrete Logarithms on a Quantum Computer}, SIAM journal on computing 26.5 : 1484-1509 (1997).
		
		
\bibitem{A fast quantum mechanical algorithm for database search}
	Lov. K. Grover, \emph{A fast quantum mechanical algorithm for database search}, Proceedings of the twenty-eighth annual ACM symposium on Theory of computing. ACM (1996).
		
\bibitem{Quantum cryptography: Public key distribution and coin tossing}
	Bennett, Charles H., and Gilles Brassard, \emph{Quantum cryptography: Public key distribution and coin tossing}, Theoretical Computer Science 560 : 7-11 (2014).
		
\bibitem{Power of One Bit of Quantum Information}
	E. Knill and R. Laflamme, \emph{Power of One Bit of Information}, Physical Review Letters 81.25: 5672 (1998).
		
\bibitem{Hardness of Classically Simulating the One-Clean-Qubit Model}
	Tomoyuki Morimae, Keisuke Fujii and Joseph F. Fitzsimons, \emph{Hardness of Classically Simulating the One-Clean-Qubit Model}, Phys. Rev. Lett. 112, 130502 (2014).
		
\bibitem{An Introduction to Boson-Sampling}
	Gard, Bryan T., et al, \emph{An introduction to boson-sampling}. From atomic to mesoscale: The role of quantum coherence in systems of various complexities, World Scientific Publishing Co. Pte. Ltd, pp 167-192 (2015).
	
\bibitem{Temporally_Unstructured_Quantum_Computation}
	Dan Shepherd and Michael J. Bremner, \emph{Temporally Unstructured Quantum Computation}, Proc. R. Soc. A 465, 1413–1439 (2009).
		
\bibitem{Fault-tolerant computing with biased-noise superconducting qubits: a case study}
	P. Aliferis, F. Brito, D. P. DiVincenzo, J. Preskill, M. Steffen and B. M. Terhal, \emph{Fault-tolerant computing with biased-noise superconducting qubits: a case study}, New J. Phys. 11, 013061 (2009).

\bibitem{Architectures for quantum simulation showing quantum supremacy}
	J. Bermejo-Vega, D. Hangleiter, M. Schwarz, R. Raussendorf and J. Eisert, \emph{Architectures for quantum simulation showing quantum supremacy}, arXiv:1703.00466 [quant-ph] (2017).
		
\bibitem{Classical Simulation of Commuting Quantum Computations Implies Collapse of the Polynomial Hierarchy}
	Michael J. Bremner, Richard Jozsa and Dan J. Shepherd, \emph{Classical Simulation of Commuting Quantum Computations Implies Collapse of the Polynomial Hierarchy}, Proc. R. Soc. A 467, 459--472 (2010).

\bibitem{Average-case complexity versus approximate simulation of commuting quantum computations}
	Bremner, Michael J., Ashley Montanaro, and Dan J. Shepherd, \emph{Average-case complexity versus approximate simulation of commuting quantum computations}, Physical review letters 117.8: 080501 (2016).
	
\bibitem{Achieving quantum supremacy with sparse and noisy commuting quantum computations}
	Bremner, Michael J., Ashley Montanaro, and Dan J. Shepherd, \emph{Achieving quantum supremacy with sparse and noisy commuting quantum computations}, arXiv preprint arXiv:1610.01808 (2016).
		
\bibitem{A one-way quantum computer}
	Raussendorf, Robert, and Hans J. Briegel, \emph{A one-way quantum computer}, Physical Review Letters 86.22: 5188 (2001).
	
\bibitem{Measurement-based quantum computation on cluster states}
	Raussendorf, Robert, Daniel E. Browne, and Hans J. Briegel, \emph{Measurement-based quantum computation on cluster states}, Physical review A 68, 022312 (2003).
	
\bibitem{Universal Blind Quantum Computation}
	Broadbent, Anne, Joseph Fitzsimons, and Elham Kashefi, \emph{Universal Blind Quantum Computation}, 50th Annual IEEE Symposium on Foundations of Computer Science ( 2009).
		
\bibitem{Unconditionally Verifiable Blind Quantum Computation}
	Joseph F. Fitzsimons and Elham Kashefi, \emph{Unconditionally Verifiable Blind Quantum Computation}, arXiv preprint arXiv:1203.5217 (2012).
	
\bibitem{Measurement-based classical computation}
	Matty J. Hoban, Joel J. Wallman, Hussain Anwar, Naïri Usher, Robert Raussendorf, Dan E. Browne, \emph{Measurement-based classical computation}, Phys. Rev. Lett. 112, 140505 (2014).


\bibitem{Multi-party entanglement in graph states}
	M. Hein, J. Eisert and H.J. Briegel, \emph{Multi-party entanglement in graph states}, Phys. Rev. A 69, 062311 (2004).
    
	

		
%
%
%
%
%
%
%
		

	
	
%
 \bibitem{Abstract cryptography}
 	Maurer, Ueli, and Renato Renner, \emph{Abstract cryptography}, In Innovations in Computer Science, Tsinghua University Press (2011).
 		
 \bibitem{Cryptographic security of quantum key distribution}
 	Portmann, Christopher, and Renato Renner, \emph{Cryptographic security of quantum key distribution}, arXiv preprint arXiv:1409.3525 (2014).
 		
\bibitem{Composable security of delegated quantum computation}
 	Dunjko, Vedran, et al, \emph{Composable security of delegated quantum computation}, International Conference on the Theory and Application of Cryptology and Information Security. Springer Berlin Heidelberg (2014).
%
%
%
\bibitem{Matroid Theory}
	J. G. Oxley, \emph{Matroid Theory}, Oxford University Press, 2011.
%
%
%
%
%
%
%
%
\end{thebibliography}
\end{document}